\documentclass{aastex62}

\usepackage{amsmath,amssymb,amsthm}

\usepackage{ulem}
\usepackage{multirow}
 
\newcommand{\xobs}{\ensuremath{\mathbf{x}_{\text{obs}}}}

\newcommand{\x}{\mathbf{x}}
\newcommand{\y}{\mathbf{y}}

\newcommand{\norm}[1]{\left\lVert#1\right\rVert}
\newcommand{\abs}[1]{\left|#1\right|}
\newcommand{\pr}[1]{\mathrm{p}(#1)} 
\newcommand{\gvn}{|}
\newcommand{\betavec}{\boldsymbol{\beta}}
\renewcommand{\hat}[1]{\widehat{#1}}

\newcommand{\code}[1]{\texttt{#1}}
\newcommand{\python}{\code{Python}}
\newcommand{\rlang}{\code{R}}
\newcommand{\nnkcde}{\code{NNKCDE}}
\newcommand{\rfcde}{\code{RFCDE}}
\newcommand{\frfcde}{\code{fRFCDE}}
\newcommand{\flexcode}{\code{FlexCode}}
\newcommand{\flexzboost}{\code{FlexZBoost}}
\newcommand{\deepcde}{\code{DeepCDE}}
\newcommand{\cdetools}{\code{cdetools}}

\newcommand{\pz}{photo-$z$}
\newcommand{\Pz}{Photo-$z$}
\newcommand{\ptr}{p_{{\rm spec}}}
\newcommand{\pte}{p_{{\rm photo}}}

\newcommand{\data}[1]{\textsc{#1}}
\newcommand{\lsst}{\data{LSST}}
\newcommand{\sdss}{\data{SDSS}}
\newcommand{\teddy}{\data{Teddy}}

\newcommand{\xobsph}{\ensuremath{\mathbf{x}_{\text{val}}}}
\newcommand{\yobsph}{\ensuremath{y_{\text{val}}}}
\newcommand{\Yobsph}{\ensuremath{\mathbf{y}_{\text{val}}}}
\newcommand{\zobsph}{\ensuremath{z_{\text{val}}}}
\newcommand{\etaobsph}{\ensuremath{\eta_{\text{val}}}}

\newtheorem{Lemma}{Lemma}

\usepackage{listings,newtxtt}
\usepackage[draft]{minted}

\usepackage{tikz}
\usetikzlibrary{calc} 
\usetikzlibrary{positioning}
\usetikzlibrary{decorations.text}
\usetikzlibrary{decorations.pathmorphing}

\makeatletter
\def\PYGdefault@reset{\let\PYGdefault@it=\relax \let\PYGdefault@bf=\relax%
    \let\PYGdefault@ul=\relax \let\PYGdefault@tc=\relax%
    \let\PYGdefault@bc=\relax \let\PYGdefault@ff=\relax}
\def\PYGdefault@tok#1{\csname PYGdefault@tok@#1\endcsname}
\def\PYGdefault@toks#1+{\ifx\relax#1\empty\else%
    \PYGdefault@tok{#1}\expandafter\PYGdefault@toks\fi}
\def\PYGdefault@do#1{\PYGdefault@bc{\PYGdefault@tc{\PYGdefault@ul{%
    \PYGdefault@it{\PYGdefault@bf{\PYGdefault@ff{#1}}}}}}}
\def\PYGdefault#1#2{\PYGdefault@reset\PYGdefault@toks#1+\relax+\PYGdefault@do{#2}}

\expandafter\def\csname PYGdefault@tok@w\endcsname{\def\PYGdefault@tc##1{\textcolor[rgb]{0.73,0.73,0.73}{##1}}}
\expandafter\def\csname PYGdefault@tok@c\endcsname{\let\PYGdefault@it=\textit\def\PYGdefault@tc##1{\textcolor[rgb]{0.25,0.50,0.50}{##1}}}
\expandafter\def\csname PYGdefault@tok@cp\endcsname{\def\PYGdefault@tc##1{\textcolor[rgb]{0.74,0.48,0.00}{##1}}}
\expandafter\def\csname PYGdefault@tok@k\endcsname{\let\PYGdefault@bf=\textbf\def\PYGdefault@tc##1{\textcolor[rgb]{0.00,0.50,0.00}{##1}}}
\expandafter\def\csname PYGdefault@tok@kp\endcsname{\def\PYGdefault@tc##1{\textcolor[rgb]{0.00,0.50,0.00}{##1}}}
\expandafter\def\csname PYGdefault@tok@kt\endcsname{\def\PYGdefault@tc##1{\textcolor[rgb]{0.69,0.00,0.25}{##1}}}
\expandafter\def\csname PYGdefault@tok@o\endcsname{\def\PYGdefault@tc##1{\textcolor[rgb]{0.40,0.40,0.40}{##1}}}
\expandafter\def\csname PYGdefault@tok@ow\endcsname{\let\PYGdefault@bf=\textbf\def\PYGdefault@tc##1{\textcolor[rgb]{0.67,0.13,1.00}{##1}}}
\expandafter\def\csname PYGdefault@tok@nb\endcsname{\def\PYGdefault@tc##1{\textcolor[rgb]{0.00,0.50,0.00}{##1}}}
\expandafter\def\csname PYGdefault@tok@nf\endcsname{\def\PYGdefault@tc##1{\textcolor[rgb]{0.00,0.00,1.00}{##1}}}
\expandafter\def\csname PYGdefault@tok@nc\endcsname{\let\PYGdefault@bf=\textbf\def\PYGdefault@tc##1{\textcolor[rgb]{0.00,0.00,1.00}{##1}}}
\expandafter\def\csname PYGdefault@tok@nn\endcsname{\let\PYGdefault@bf=\textbf\def\PYGdefault@tc##1{\textcolor[rgb]{0.00,0.00,1.00}{##1}}}
\expandafter\def\csname PYGdefault@tok@ne\endcsname{\let\PYGdefault@bf=\textbf\def\PYGdefault@tc##1{\textcolor[rgb]{0.82,0.25,0.23}{##1}}}
\expandafter\def\csname PYGdefault@tok@nv\endcsname{\def\PYGdefault@tc##1{\textcolor[rgb]{0.10,0.09,0.49}{##1}}}
\expandafter\def\csname PYGdefault@tok@no\endcsname{\def\PYGdefault@tc##1{\textcolor[rgb]{0.53,0.00,0.00}{##1}}}
\expandafter\def\csname PYGdefault@tok@nl\endcsname{\def\PYGdefault@tc##1{\textcolor[rgb]{0.63,0.63,0.00}{##1}}}
\expandafter\def\csname PYGdefault@tok@ni\endcsname{\let\PYGdefault@bf=\textbf\def\PYGdefault@tc##1{\textcolor[rgb]{0.60,0.60,0.60}{##1}}}
\expandafter\def\csname PYGdefault@tok@na\endcsname{\def\PYGdefault@tc##1{\textcolor[rgb]{0.49,0.56,0.16}{##1}}}
\expandafter\def\csname PYGdefault@tok@nt\endcsname{\let\PYGdefault@bf=\textbf\def\PYGdefault@tc##1{\textcolor[rgb]{0.00,0.50,0.00}{##1}}}
\expandafter\def\csname PYGdefault@tok@nd\endcsname{\def\PYGdefault@tc##1{\textcolor[rgb]{0.67,0.13,1.00}{##1}}}
\expandafter\def\csname PYGdefault@tok@s\endcsname{\def\PYGdefault@tc##1{\textcolor[rgb]{0.73,0.13,0.13}{##1}}}
\expandafter\def\csname PYGdefault@tok@sd\endcsname{\let\PYGdefault@it=\textit\def\PYGdefault@tc##1{\textcolor[rgb]{0.73,0.13,0.13}{##1}}}
\expandafter\def\csname PYGdefault@tok@si\endcsname{\let\PYGdefault@bf=\textbf\def\PYGdefault@tc##1{\textcolor[rgb]{0.73,0.40,0.53}{##1}}}
\expandafter\def\csname PYGdefault@tok@se\endcsname{\let\PYGdefault@bf=\textbf\def\PYGdefault@tc##1{\textcolor[rgb]{0.73,0.40,0.13}{##1}}}
\expandafter\def\csname PYGdefault@tok@sr\endcsname{\def\PYGdefault@tc##1{\textcolor[rgb]{0.73,0.40,0.53}{##1}}}
\expandafter\def\csname PYGdefault@tok@ss\endcsname{\def\PYGdefault@tc##1{\textcolor[rgb]{0.10,0.09,0.49}{##1}}}
\expandafter\def\csname PYGdefault@tok@sx\endcsname{\def\PYGdefault@tc##1{\textcolor[rgb]{0.00,0.50,0.00}{##1}}}
\expandafter\def\csname PYGdefault@tok@m\endcsname{\def\PYGdefault@tc##1{\textcolor[rgb]{0.40,0.40,0.40}{##1}}}
\expandafter\def\csname PYGdefault@tok@gh\endcsname{\let\PYGdefault@bf=\textbf\def\PYGdefault@tc##1{\textcolor[rgb]{0.00,0.00,0.50}{##1}}}
\expandafter\def\csname PYGdefault@tok@gu\endcsname{\let\PYGdefault@bf=\textbf\def\PYGdefault@tc##1{\textcolor[rgb]{0.50,0.00,0.50}{##1}}}
\expandafter\def\csname PYGdefault@tok@gd\endcsname{\def\PYGdefault@tc##1{\textcolor[rgb]{0.63,0.00,0.00}{##1}}}
\expandafter\def\csname PYGdefault@tok@gi\endcsname{\def\PYGdefault@tc##1{\textcolor[rgb]{0.00,0.63,0.00}{##1}}}
\expandafter\def\csname PYGdefault@tok@gr\endcsname{\def\PYGdefault@tc##1{\textcolor[rgb]{1.00,0.00,0.00}{##1}}}
\expandafter\def\csname PYGdefault@tok@ge\endcsname{\let\PYGdefault@it=\textit}
\expandafter\def\csname PYGdefault@tok@gs\endcsname{\let\PYGdefault@bf=\textbf}
\expandafter\def\csname PYGdefault@tok@gp\endcsname{\let\PYGdefault@bf=\textbf\def\PYGdefault@tc##1{\textcolor[rgb]{0.00,0.00,0.50}{##1}}}
\expandafter\def\csname PYGdefault@tok@go\endcsname{\def\PYGdefault@tc##1{\textcolor[rgb]{0.53,0.53,0.53}{##1}}}
\expandafter\def\csname PYGdefault@tok@gt\endcsname{\def\PYGdefault@tc##1{\textcolor[rgb]{0.00,0.27,0.87}{##1}}}
\expandafter\def\csname PYGdefault@tok@err\endcsname{\def\PYGdefault@bc##1{\setlength{\fboxsep}{0pt}\fcolorbox[rgb]{1.00,0.00,0.00}{1,1,1}{\strut ##1}}}
\expandafter\def\csname PYGdefault@tok@kc\endcsname{\let\PYGdefault@bf=\textbf\def\PYGdefault@tc##1{\textcolor[rgb]{0.00,0.50,0.00}{##1}}}
\expandafter\def\csname PYGdefault@tok@kd\endcsname{\let\PYGdefault@bf=\textbf\def\PYGdefault@tc##1{\textcolor[rgb]{0.00,0.50,0.00}{##1}}}
\expandafter\def\csname PYGdefault@tok@kn\endcsname{\let\PYGdefault@bf=\textbf\def\PYGdefault@tc##1{\textcolor[rgb]{0.00,0.50,0.00}{##1}}}
\expandafter\def\csname PYGdefault@tok@kr\endcsname{\let\PYGdefault@bf=\textbf\def\PYGdefault@tc##1{\textcolor[rgb]{0.00,0.50,0.00}{##1}}}
\expandafter\def\csname PYGdefault@tok@bp\endcsname{\def\PYGdefault@tc##1{\textcolor[rgb]{0.00,0.50,0.00}{##1}}}
\expandafter\def\csname PYGdefault@tok@fm\endcsname{\def\PYGdefault@tc##1{\textcolor[rgb]{0.00,0.00,1.00}{##1}}}
\expandafter\def\csname PYGdefault@tok@vc\endcsname{\def\PYGdefault@tc##1{\textcolor[rgb]{0.10,0.09,0.49}{##1}}}
\expandafter\def\csname PYGdefault@tok@vg\endcsname{\def\PYGdefault@tc##1{\textcolor[rgb]{0.10,0.09,0.49}{##1}}}
\expandafter\def\csname PYGdefault@tok@vi\endcsname{\def\PYGdefault@tc##1{\textcolor[rgb]{0.10,0.09,0.49}{##1}}}
\expandafter\def\csname PYGdefault@tok@vm\endcsname{\def\PYGdefault@tc##1{\textcolor[rgb]{0.10,0.09,0.49}{##1}}}
\expandafter\def\csname PYGdefault@tok@sa\endcsname{\def\PYGdefault@tc##1{\textcolor[rgb]{0.73,0.13,0.13}{##1}}}
\expandafter\def\csname PYGdefault@tok@sb\endcsname{\def\PYGdefault@tc##1{\textcolor[rgb]{0.73,0.13,0.13}{##1}}}
\expandafter\def\csname PYGdefault@tok@sc\endcsname{\def\PYGdefault@tc##1{\textcolor[rgb]{0.73,0.13,0.13}{##1}}}
\expandafter\def\csname PYGdefault@tok@dl\endcsname{\def\PYGdefault@tc##1{\textcolor[rgb]{0.73,0.13,0.13}{##1}}}
\expandafter\def\csname PYGdefault@tok@s2\endcsname{\def\PYGdefault@tc##1{\textcolor[rgb]{0.73,0.13,0.13}{##1}}}
\expandafter\def\csname PYGdefault@tok@sh\endcsname{\def\PYGdefault@tc##1{\textcolor[rgb]{0.73,0.13,0.13}{##1}}}
\expandafter\def\csname PYGdefault@tok@s1\endcsname{\def\PYGdefault@tc##1{\textcolor[rgb]{0.73,0.13,0.13}{##1}}}
\expandafter\def\csname PYGdefault@tok@mb\endcsname{\def\PYGdefault@tc##1{\textcolor[rgb]{0.40,0.40,0.40}{##1}}}
\expandafter\def\csname PYGdefault@tok@mf\endcsname{\def\PYGdefault@tc##1{\textcolor[rgb]{0.40,0.40,0.40}{##1}}}
\expandafter\def\csname PYGdefault@tok@mh\endcsname{\def\PYGdefault@tc##1{\textcolor[rgb]{0.40,0.40,0.40}{##1}}}
\expandafter\def\csname PYGdefault@tok@mi\endcsname{\def\PYGdefault@tc##1{\textcolor[rgb]{0.40,0.40,0.40}{##1}}}
\expandafter\def\csname PYGdefault@tok@il\endcsname{\def\PYGdefault@tc##1{\textcolor[rgb]{0.40,0.40,0.40}{##1}}}
\expandafter\def\csname PYGdefault@tok@mo\endcsname{\def\PYGdefault@tc##1{\textcolor[rgb]{0.40,0.40,0.40}{##1}}}
\expandafter\def\csname PYGdefault@tok@ch\endcsname{\let\PYGdefault@it=\textit\def\PYGdefault@tc##1{\textcolor[rgb]{0.25,0.50,0.50}{##1}}}
\expandafter\def\csname PYGdefault@tok@cm\endcsname{\let\PYGdefault@it=\textit\def\PYGdefault@tc##1{\textcolor[rgb]{0.25,0.50,0.50}{##1}}}
\expandafter\def\csname PYGdefault@tok@cpf\endcsname{\let\PYGdefault@it=\textit\def\PYGdefault@tc##1{\textcolor[rgb]{0.25,0.50,0.50}{##1}}}
\expandafter\def\csname PYGdefault@tok@c1\endcsname{\let\PYGdefault@it=\textit\def\PYGdefault@tc##1{\textcolor[rgb]{0.25,0.50,0.50}{##1}}}
\expandafter\def\csname PYGdefault@tok@cs\endcsname{\let\PYGdefault@it=\textit\def\PYGdefault@tc##1{\textcolor[rgb]{0.25,0.50,0.50}{##1}}}

\makeatother

\makeatletter
\def\PYG@reset{\let\PYG@it=\relax \let\PYG@bf=\relax%
    \let\PYG@ul=\relax \let\PYG@tc=\relax%
    \let\PYG@bc=\relax \let\PYG@ff=\relax}
\def\PYG@tok#1{\csname PYG@tok@#1\endcsname}
\def\PYG@toks#1+{\ifx\relax#1\empty\else%
    \PYG@tok{#1}\expandafter\PYG@toks\fi}
\def\PYG@do#1{\PYG@bc{\PYG@tc{\PYG@ul{%
    \PYG@it{\PYG@bf{\PYG@ff{#1}}}}}}}
\def\PYG#1#2{\PYG@reset\PYG@toks#1+\relax+\PYG@do{#2}}

\expandafter\def\csname PYG@tok@w\endcsname{\def\PYG@tc##1{\textcolor[rgb]{0.73,0.73,0.73}{##1}}}
\expandafter\def\csname PYG@tok@c\endcsname{\let\PYG@it=\textit\def\PYG@tc##1{\textcolor[rgb]{0.25,0.50,0.50}{##1}}}
\expandafter\def\csname PYG@tok@cp\endcsname{\def\PYG@tc##1{\textcolor[rgb]{0.74,0.48,0.00}{##1}}}
\expandafter\def\csname PYG@tok@k\endcsname{\let\PYG@bf=\textbf\def\PYG@tc##1{\textcolor[rgb]{0.00,0.50,0.00}{##1}}}
\expandafter\def\csname PYG@tok@kp\endcsname{\def\PYG@tc##1{\textcolor[rgb]{0.00,0.50,0.00}{##1}}}
\expandafter\def\csname PYG@tok@kt\endcsname{\def\PYG@tc##1{\textcolor[rgb]{0.69,0.00,0.25}{##1}}}
\expandafter\def\csname PYG@tok@o\endcsname{\def\PYG@tc##1{\textcolor[rgb]{0.40,0.40,0.40}{##1}}}
\expandafter\def\csname PYG@tok@ow\endcsname{\let\PYG@bf=\textbf\def\PYG@tc##1{\textcolor[rgb]{0.67,0.13,1.00}{##1}}}
\expandafter\def\csname PYG@tok@nb\endcsname{\def\PYG@tc##1{\textcolor[rgb]{0.00,0.50,0.00}{##1}}}
\expandafter\def\csname PYG@tok@nf\endcsname{\def\PYG@tc##1{\textcolor[rgb]{0.00,0.00,1.00}{##1}}}
\expandafter\def\csname PYG@tok@nc\endcsname{\let\PYG@bf=\textbf\def\PYG@tc##1{\textcolor[rgb]{0.00,0.00,1.00}{##1}}}
\expandafter\def\csname PYG@tok@nn\endcsname{\let\PYG@bf=\textbf\def\PYG@tc##1{\textcolor[rgb]{0.00,0.00,1.00}{##1}}}
\expandafter\def\csname PYG@tok@ne\endcsname{\let\PYG@bf=\textbf\def\PYG@tc##1{\textcolor[rgb]{0.82,0.25,0.23}{##1}}}
\expandafter\def\csname PYG@tok@nv\endcsname{\def\PYG@tc##1{\textcolor[rgb]{0.10,0.09,0.49}{##1}}}
\expandafter\def\csname PYG@tok@no\endcsname{\def\PYG@tc##1{\textcolor[rgb]{0.53,0.00,0.00}{##1}}}
\expandafter\def\csname PYG@tok@nl\endcsname{\def\PYG@tc##1{\textcolor[rgb]{0.63,0.63,0.00}{##1}}}
\expandafter\def\csname PYG@tok@ni\endcsname{\let\PYG@bf=\textbf\def\PYG@tc##1{\textcolor[rgb]{0.60,0.60,0.60}{##1}}}
\expandafter\def\csname PYG@tok@na\endcsname{\def\PYG@tc##1{\textcolor[rgb]{0.49,0.56,0.16}{##1}}}
\expandafter\def\csname PYG@tok@nt\endcsname{\let\PYG@bf=\textbf\def\PYG@tc##1{\textcolor[rgb]{0.00,0.50,0.00}{##1}}}
\expandafter\def\csname PYG@tok@nd\endcsname{\def\PYG@tc##1{\textcolor[rgb]{0.67,0.13,1.00}{##1}}}
\expandafter\def\csname PYG@tok@s\endcsname{\def\PYG@tc##1{\textcolor[rgb]{0.73,0.13,0.13}{##1}}}
\expandafter\def\csname PYG@tok@sd\endcsname{\let\PYG@it=\textit\def\PYG@tc##1{\textcolor[rgb]{0.73,0.13,0.13}{##1}}}
\expandafter\def\csname PYG@tok@si\endcsname{\let\PYG@bf=\textbf\def\PYG@tc##1{\textcolor[rgb]{0.73,0.40,0.53}{##1}}}
\expandafter\def\csname PYG@tok@se\endcsname{\let\PYG@bf=\textbf\def\PYG@tc##1{\textcolor[rgb]{0.73,0.40,0.13}{##1}}}
\expandafter\def\csname PYG@tok@sr\endcsname{\def\PYG@tc##1{\textcolor[rgb]{0.73,0.40,0.53}{##1}}}
\expandafter\def\csname PYG@tok@ss\endcsname{\def\PYG@tc##1{\textcolor[rgb]{0.10,0.09,0.49}{##1}}}
\expandafter\def\csname PYG@tok@sx\endcsname{\def\PYG@tc##1{\textcolor[rgb]{0.00,0.50,0.00}{##1}}}
\expandafter\def\csname PYG@tok@m\endcsname{\def\PYG@tc##1{\textcolor[rgb]{0.40,0.40,0.40}{##1}}}
\expandafter\def\csname PYG@tok@gh\endcsname{\let\PYG@bf=\textbf\def\PYG@tc##1{\textcolor[rgb]{0.00,0.00,0.50}{##1}}}
\expandafter\def\csname PYG@tok@gu\endcsname{\let\PYG@bf=\textbf\def\PYG@tc##1{\textcolor[rgb]{0.50,0.00,0.50}{##1}}}
\expandafter\def\csname PYG@tok@gd\endcsname{\def\PYG@tc##1{\textcolor[rgb]{0.63,0.00,0.00}{##1}}}
\expandafter\def\csname PYG@tok@gi\endcsname{\def\PYG@tc##1{\textcolor[rgb]{0.00,0.63,0.00}{##1}}}
\expandafter\def\csname PYG@tok@gr\endcsname{\def\PYG@tc##1{\textcolor[rgb]{1.00,0.00,0.00}{##1}}}
\expandafter\def\csname PYG@tok@ge\endcsname{\let\PYG@it=\textit}
\expandafter\def\csname PYG@tok@gs\endcsname{\let\PYG@bf=\textbf}
\expandafter\def\csname PYG@tok@gp\endcsname{\let\PYG@bf=\textbf\def\PYG@tc##1{\textcolor[rgb]{0.00,0.00,0.50}{##1}}}
\expandafter\def\csname PYG@tok@go\endcsname{\def\PYG@tc##1{\textcolor[rgb]{0.53,0.53,0.53}{##1}}}
\expandafter\def\csname PYG@tok@gt\endcsname{\def\PYG@tc##1{\textcolor[rgb]{0.00,0.27,0.87}{##1}}}
\expandafter\def\csname PYG@tok@err\endcsname{\def\PYG@bc##1{\setlength{\fboxsep}{0pt}\fcolorbox[rgb]{1.00,0.00,0.00}{1,1,1}{\strut ##1}}}
\expandafter\def\csname PYG@tok@kc\endcsname{\let\PYG@bf=\textbf\def\PYG@tc##1{\textcolor[rgb]{0.00,0.50,0.00}{##1}}}
\expandafter\def\csname PYG@tok@kd\endcsname{\let\PYG@bf=\textbf\def\PYG@tc##1{\textcolor[rgb]{0.00,0.50,0.00}{##1}}}
\expandafter\def\csname PYG@tok@kn\endcsname{\let\PYG@bf=\textbf\def\PYG@tc##1{\textcolor[rgb]{0.00,0.50,0.00}{##1}}}
\expandafter\def\csname PYG@tok@kr\endcsname{\let\PYG@bf=\textbf\def\PYG@tc##1{\textcolor[rgb]{0.00,0.50,0.00}{##1}}}
\expandafter\def\csname PYG@tok@bp\endcsname{\def\PYG@tc##1{\textcolor[rgb]{0.00,0.50,0.00}{##1}}}
\expandafter\def\csname PYG@tok@fm\endcsname{\def\PYG@tc##1{\textcolor[rgb]{0.00,0.00,1.00}{##1}}}
\expandafter\def\csname PYG@tok@vc\endcsname{\def\PYG@tc##1{\textcolor[rgb]{0.10,0.09,0.49}{##1}}}
\expandafter\def\csname PYG@tok@vg\endcsname{\def\PYG@tc##1{\textcolor[rgb]{0.10,0.09,0.49}{##1}}}
\expandafter\def\csname PYG@tok@vi\endcsname{\def\PYG@tc##1{\textcolor[rgb]{0.10,0.09,0.49}{##1}}}
\expandafter\def\csname PYG@tok@vm\endcsname{\def\PYG@tc##1{\textcolor[rgb]{0.10,0.09,0.49}{##1}}}
\expandafter\def\csname PYG@tok@sa\endcsname{\def\PYG@tc##1{\textcolor[rgb]{0.73,0.13,0.13}{##1}}}
\expandafter\def\csname PYG@tok@sb\endcsname{\def\PYG@tc##1{\textcolor[rgb]{0.73,0.13,0.13}{##1}}}
\expandafter\def\csname PYG@tok@sc\endcsname{\def\PYG@tc##1{\textcolor[rgb]{0.73,0.13,0.13}{##1}}}
\expandafter\def\csname PYG@tok@dl\endcsname{\def\PYG@tc##1{\textcolor[rgb]{0.73,0.13,0.13}{##1}}}
\expandafter\def\csname PYG@tok@s2\endcsname{\def\PYG@tc##1{\textcolor[rgb]{0.73,0.13,0.13}{##1}}}
\expandafter\def\csname PYG@tok@sh\endcsname{\def\PYG@tc##1{\textcolor[rgb]{0.73,0.13,0.13}{##1}}}
\expandafter\def\csname PYG@tok@s1\endcsname{\def\PYG@tc##1{\textcolor[rgb]{0.73,0.13,0.13}{##1}}}
\expandafter\def\csname PYG@tok@mb\endcsname{\def\PYG@tc##1{\textcolor[rgb]{0.40,0.40,0.40}{##1}}}
\expandafter\def\csname PYG@tok@mf\endcsname{\def\PYG@tc##1{\textcolor[rgb]{0.40,0.40,0.40}{##1}}}
\expandafter\def\csname PYG@tok@mh\endcsname{\def\PYG@tc##1{\textcolor[rgb]{0.40,0.40,0.40}{##1}}}
\expandafter\def\csname PYG@tok@mi\endcsname{\def\PYG@tc##1{\textcolor[rgb]{0.40,0.40,0.40}{##1}}}
\expandafter\def\csname PYG@tok@il\endcsname{\def\PYG@tc##1{\textcolor[rgb]{0.40,0.40,0.40}{##1}}}
\expandafter\def\csname PYG@tok@mo\endcsname{\def\PYG@tc##1{\textcolor[rgb]{0.40,0.40,0.40}{##1}}}
\expandafter\def\csname PYG@tok@ch\endcsname{\let\PYG@it=\textit\def\PYG@tc##1{\textcolor[rgb]{0.25,0.50,0.50}{##1}}}
\expandafter\def\csname PYG@tok@cm\endcsname{\let\PYG@it=\textit\def\PYG@tc##1{\textcolor[rgb]{0.25,0.50,0.50}{##1}}}
\expandafter\def\csname PYG@tok@cpf\endcsname{\let\PYG@it=\textit\def\PYG@tc##1{\textcolor[rgb]{0.25,0.50,0.50}{##1}}}
\expandafter\def\csname PYG@tok@c1\endcsname{\let\PYG@it=\textit\def\PYG@tc##1{\textcolor[rgb]{0.25,0.50,0.50}{##1}}}
\expandafter\def\csname PYG@tok@cs\endcsname{\let\PYG@it=\textit\def\PYG@tc##1{\textcolor[rgb]{0.25,0.50,0.50}{##1}}}

\def\PYGZus{\char`\_}
\def\PYGZob{\char`\{}
\def\PYGZcb{\char`\}}

\def\PYGZlt{\char`\<}

\def\PYGZsh{\char`\#}

\def\PYGZdl{\char`\$}
\def\PYGZhy{\char`\-}
\def\PYGZsq{\char`\'}
\def\PYGZdq{\char`\"}

\makeatother

\begin{document}

\shorttitle{Tools for Conditional Density Estimation}
\title{Conditional Density Estimation Tools in \texttt{Python} and \texttt{R}\\
with Applications to Photometric Redshifts and Likelihood-Free Cosmological Inference}

\shortauthors{Dalmasso et al.}

\author{Niccol\`o Dalmasso}
\affiliation{Department of Statistics \& Data Science, Carnegie Mellon University, USA}

\correspondingauthor{Niccol\`o Dalmasso, \texttt{ndalmass@stat.cmu.edu}}

\author{Taylor Pospisil}
\affiliation{Google LLC, Mountain View, USA} 

\author{Ann B. Lee}
\affiliation{Department of Statistics \& Data Science, Carnegie Mellon University, USA}

\author{Rafael Izbicki}
\affiliation{Department of Statistics, Federal University of S\~ao Carlos, Brazil}

\author{Peter E. Freeman}
\affiliation{Department of Statistics \& Data Science, Carnegie Mellon University, USA}

\author{Alex I. Malz}
\affiliation{German Centre of Cosmological Lensing, Ruhr-Universit\"{a}t Bochum, Germany}
\affiliation{Center for Cosmology and Particle Physics, New York University, USA}

\begin{abstract}
It is well known in astronomy that propagating non-Gaussian prediction uncertainty in photometric redshift estimates is key to reducing bias in downstream cosmological analyses.
Similarly, likelihood-free inference approaches, which are beginning to emerge as a tool for cosmological analysis, require a characterization of the full uncertainty landscape of the parameters of interest given observed data.
However, most machine learning (ML) or training-based methods with open-source software target point prediction or classification, and hence fall short in quantifying uncertainty in complex regression and parameter inference settings such as the applications mentioned above.  
As an alternative to methods that focus on predicting the response (or parameters) $\y$ from features $\x$, we provide nonparametric conditional density estimation (CDE) tools for approximating and validating the entire probability density function (PDF)  $\mathrm{p}(\y | \x)$ of $\y$ given (i.e., conditional on) $\x$.
This density approach offers a more nuanced accounting of uncertainty in situations with, e.g., nonstandard error distributions and multimodal or heteroskedastic response variables that are often present in astronomical data sets. 
As there is no one-size-fits-all CDE method, and the ultimate choice of model depends on the application and the training sample size, the goal of this work is to provide a comprehensive range of statistical tools and open-source software for nonparametric CDE and method assessment which can accommodate different types of settings --- involving, e.g., mixed-type input from multiple sources, functional data, and images --- and which in addition can easily be fit to the problem at hand. 
Specifically, we introduce four CDE software packages in \texttt{Python} and \texttt{R} based on ML prediction methods  \emph{adapted and optimized for CDE}:
\texttt{NNKCDE}, \texttt{RFCDE}, \texttt{FlexCode}, and \texttt{DeepCDE}.
Furthermore, we present the \texttt{cdetools} package with evaluation metrics. 
This package includes functions for computing a CDE loss function for tuning and assessing the quality of individual PDFs, together with diagnostic functions that probe the population-level performance of the PDFs.  
We provide sample code in \texttt{Python} and \texttt{R} as well as examples of applications to photometric redshift estimation and likelihood-free cosmological inference via CDE.
\end{abstract}


\section{Introduction}

Machine learning (ML) has seen a surge in popularity in almost all fields of astronomy that involve massive amounts of complex data \citep{ball2010data, way2012advances, ivezic2014statistics, ntampaka2019role}.
Most ML methods target regression and classification, whose primary goal is to return a point estimate of an unknown response variable $\y$ given observed features $\x$, often falling short in quantifying nontrivial uncertainty in $\y$.
For instance, returning a point estimate for a supernova's type $\y$ given a supernova's light curve $\x$, or for a galaxy mass $\y$ given its light spectrum $\x$, fails to capture degeneracies in the mapping from $\x$ to $\y$.
Neglecting to propagate these uncertainties through down-stream  analyses may lead to imprecise or even inaccurate inferences of physical parameters. 

Consider the following two examples of problems where uncertainty quantification can be impactful:
\begin{itemize}
    \item \textbf{Photometric redshift estimation.}
    In photometric redshift (\pz) estimation, one attempts to constrain
    the cosmological {\em redshift} ($z$) of a galaxy  after observing the shifted spectrum using a handful of broadband filters, and sometimes additional variables such as morphology and environment.
    In a prediction setting, the response $\y$ could be the galaxy's true (i.e. spectroscopically observed) redshift but could also include other galaxy properties ($\alpha$) such as the galaxy's mass or age; 
    the features $\x$ would represent the collection of directly observable inputs such as photometric magnitudes and colors used to predict $\y = z$ or more generally
    a multivariate response $\y = (z, \alpha)$. 
    The use of \pz\ posterior estimates --- that is, estimates of the probability density functions (PDFs) for individual galaxies --- is crucial for cosmological inference from photometric galaxy surveys, as $\gtrsim 99\%$ of currently cataloged galaxies are observed solely via photometry, a percentage that will only grow in the coming decade as the Large Synoptic Survey Telescope (\lsst) begins gathering data \citep{Ivezic19}. 
    In addition, widely different redshifts can be consistent with the observed colors of a galaxy; 
    \pz\ posterior estimates can capture such degeneracy or multimodality in the distribution whereas point estimates cannot.
    Though the benefits of using posteriors over $\alpha$ have yet to be fully exploited \citep{2018A&A...614A.129V}, it is thoroughly established that one can improve down-stream cosmological analysis by properly propagating \pz\ estimate uncertainties via probability density functions (PDFs) rather than just using simple point predictions of $\y$ \citep{mandelbaum08precisionphotoz, wittman2009lies, sheldon2012photometric, carrasco2013tpz, graff2014skynet, desc_photoz}. 
    \item \textbf{Forward-modeled observables in cosmology.} 
    Some cosmological probes, such as the type Ia supernova (SN Ia) distance-redshift relationship, have observables that are straightforward to simulate in spite of an intractable likelihood.
    Likelihood-Free Inference (LFI) methods allow for parameter inference in settings where the likelihood function, which relates the observable  data $\mathbf{x}_{\rm obs}$ to the parameters of interest $\boldsymbol{\theta}$, is too complex to work with directly, but one is able to {\em simulate} $\mathbf{x}$ from a stochastic forward model at fixed parameter settings $\boldsymbol{\theta}$. 
    The most common approach to LFI or simulation-based parameter inference is Approximate Bayesian Computation (ABC), whose many variations (see \citealt{beaumont2002approximate} and \citealt{sisson2018handbook} for a review) repeatedly draw from the simulator and compare the output with observed data $\xobs$ in real time to arrive at a set of plausible parameter values consistent with $\xobs$.
    With computationally intensive simulations, however, a classical ABC approach may not be practical. 
    An alternative approach to ABC rejection sampling is to use faster training-based methods to assess the uncertainty about $\boldsymbol{\theta}$ for any $\x$ first, and then consider the specific case $\x = \xobs$.
\end{itemize}

From a statistical perspective, the right tool for quantifying the uncertainty about $\y$ once $\x$ is observed is the {\em conditional density} $\pr{\mathbf{y} \gvn \mathbf{x}}$.
In a prediction context such as for \pz\ problems, where heteroskedastic errors or multimodal responses may occur, conditional density estimation (CDE) of the density $\pr{z \gvn \mathbf{x}}$ for the redshift $z$ of individual galaxies given photometric data $\mathbf{x}$ provides a more nuanced accounting of uncertainty than a point estimate or prediction interval alone. 
CDE can also be used in LFI where, in our notation, the parameters of interest $\boldsymbol{\theta}$ take the role of the ``response'' $y$.
In such settings, one can apply training-based approaches to forward-simulated data to estimate the posterior probability distribution $\pr{\boldsymbol{\theta} \gvn \xobs}$ of cosmological parameters $\boldsymbol{\theta}$ given observed data $\xobs$, and from these posteriors then derive, e.g. posterior credible intervals of $\theta$.
Section \ref{sec:mult_response_mult_data} shows an example of cosmological parameter inference using mock weak lensing data.
Here we follow \citet{izbicki2014high, izbicki2019abc} to combine ABC and CDE by directly applying CDE techniques (Section \ref{sec:cde_methods}) and loss functions (Section \ref{sec:cde_loss}) to simulated data $\{(\boldsymbol{\theta}_i, \mathbf{x}_i)\}_{i=1}^{n}$.
Similar works include performing LFI via CDE using Gaussian copulas \citep{li2017gaussiancopula, chen2019gaussiancopula} and random forests \citep{marin2016abcrf}. 
Other examples include neural density estimation in LFI via mixture density networks and masked autoregressive flows \citep{papamakarios2016fast, Lueckmann2017posterior, lueckmann2019likelihood, papamakarios2019sequential, greenberg2019automatic, alsing2018massive, alsing2019fast}.

Data in astronomy present a challenge to estimating conditional densities, due to both the complexity of the underlying physical processes and the complicated observational noise. 
Precision cosmology, for example, requires combining data from different scientific probes, {each affected by unique sources of systematic uncertainty, to produce samples from} complicated joint likelihood functions {with nontrivial covariances} in a high-dimensional parameter space \citep{krause2017dark, Joudaki17, aghanim2018planck, vanUitert18, Abbott18}. 
In such situations, CDE methods that target a variety of settings and non-standard data (images, correlation functions, mixed data types) become especially valuable. 
However, for any given data type, there is no one-size-fits-all CDE method.
For example, deep neural networks often perform well in settings with large amounts of representative training data but in applications with smaller training samples one may need a different tool. 
There is also additional value in models that are interpretable and easy to fit to the data at hand.

{\bf The goal of this paper is to provide statistical tools and open-source software for nonparametric CDE and method assessment appropriate for challenging data in a variety of inference scenarios.} 

To our knowledge, existing CDE software either targets discrete $\y$ (e.g., probabilistic classifiers or ordinal classification \citep{eibe2001OrdinalClass}) or uses kernel density estimation (KDE) 
across all data points to provide an estimate for a continuous $\y$ \citep{hdrcde}.
What distinguishes our methodology work from others is that we are able to adapt virtually any training-based prediction method to the problem of estimating full probability distributions. 
By leveraging existing ML methods, we are hence able to provide uncertainty prediction methods and software for more general and complex data settings than `the computational tools currently available 
in the literature. 
In this paper, we showcase and provide \python\ and \rlang\ code for four flexible CDE methods: 
\nnkcde, \rfcde/\frfcde, \flexcode\ and \deepcde.\footnote{All of these methods, except for \deepcde, have occurred in previously published or archived papers. 
Hence, we only briefly review the methods in this paper and instead focus on software usage, algorithmic aspects, and the settings under which each method can be applied.} 
Each CDE approach has particular usage for different settings of response dimensionality, feature types, and computational requirements. 
Table \ref{tab:method-strenghts} provides a high-level overview in the top panel, and lists some properties of each method in the bottom panel.

A highlight of our software is that it makes uncertainty quantification straightforward for users of standard open-source machine learning \python\ packages.
As \nnkcde, \rfcde/\frfcde, \flexcode\ share the \code{sklearn} \citep{Pedregosa2011sklearn} API (with \texttt{fit} and \texttt{predict} methods), our methods are usable within the \code{sklearn} ecosystem for, e.g., cross-validation and model selection. 
In addition, \flexcode\ is essentially a plug-in method where the user can utilize
any \code{sklearn}-compatible regression function.
\deepcde\ has implementations for both \code{Tensorflow} \citep{tensorflow2015-whitepaper} and \code{Pytorch} \citep{paszke2019pytorch}, two of the most widely used deep learning frameworks.

\begin{table}
\caption{\textit{Top:} Naming convention, high-level summary and hyper-parameters of CDE methods, along with references for further details and code examples.
\textit{Bottom:} Comparison of CDE methods in terms of training capacity and compatibility with multivariate response and different types of features, with capacities estimated based on input with around 100 features and a standard i5/i7/quad-core processor with $16$GB of RAM.
Note that less complex methods (such as \nnkcde) tend to be easier to use, easier to interpret, and often perform better in settings with smaller training sets, whereas more complex methods (such as \deepcde) perform better in settings with larger (representative) training sets.}

\begin{tabular}{|c|c|l|l|c|}
\multicolumn{1}{c}{\textbf{Method}} & \multicolumn{1}{c}{\textbf{Name}}  & \multicolumn{1}{c}{\textbf{Summary}} & \multicolumn{1}{c}{\textbf{Hyper-parameters}} & \multicolumn{1}{c}{\textbf{Details}} \\ \hline
\multirow{4}{*}{\nnkcde} & & Computes a KDE estimate of & & \\ 
& \code{N}earest \code{N}eighbor & $\,$ multivariate $\y$ using the 
nearest &  $\bullet$ Number of neighbors $k$ & Section \ref{sec:NNKCDE} \\ 
& $\,$ \code{K}ernel \code{CDE} & $\,$ neighbors of the evaluation point  $\x$ & $\bullet$ Kernel bandwidth $h$ & (Code: \ref{app:nnkcde})  \\ 
& & $\,$ in feature space. &  & \\ \hline
\multirow{5}{*}{\rfcde} &  & Random forest  that partitions & &  \\ 
&  &  $\,$ the feature space using a CDE loss. 
 & $\bullet$ Random forest hyperparams. &  Section \ref{sec:RFCDE}  \\ 
& \code{R}andom \code{F}orest \code{CDE} & $\,$ 
Constructs a weighted KDE estimate & $\bullet$ Kernel bandwidth $h$ & (Code: \ref{app:rfcde}) \\ 
&  & $\,$ of multivariate $\y$ with weights & & \\
&  &  $\,$ defined by leaves in the forest. & & \\
\hline
\multirow{4}{*}{\frfcde} &  & \rfcde\ version suitable for functional
 & $\bullet$ Random forest hyperparams. &  \\ 
& \code{f}unctional \code{R}andom &  $\,$ features $\x$. Partitions the feature & $\bullet$ Kernel bandwidth $h$ & Section \ref{sec:frfcde}\\ 
& $\,$ \code{F}orest \code{CDE} & $\,$ space directly rather than & $\bullet$ Partition parameter $\lambda$ & (Code: \ref{app:frfcde}) \\ 
& & $\,$ representing $\x$ as a vector. & & \\ \hline
\multirow{3}{*}{\flexcode} & \code{Flex}ible \code{Co}nditional & Uses basis expansion of univariate $y$ & $\bullet$ Number of expansion coeffs. & Section \ref{sec:FlexCode}  \\ 
& \code{D}ensity \code{E}stimation &  $\,$ to turn CDE into a series of &  $\bullet$ Selected regression method & (Code: \ref{app:flexcode})\\ 
&  & $\,$ univariate regression problems. & $\quad$ hyperparams. & \\  \hline
\multirow{5}{*}{\deepcde} & & Uses basis expansion of univariate
 & & \\ 
 & \code{Deep} Neural & $\,$ $y$ similar to \flexcode, but learns
 & $\bullet$ Number of expansion coeffs. & Section \ref{sec:DeepCDE} \\ 
& Networks \code{CDE} &  $\,$ the expansion coefficients & $\bullet$ Selected deep neural network & (Code: \ref{sec:code_examples}) \\ 
& & $\,$ simultaneously using a deep &  $\quad$ architecture hyperparams. & \\
& & $\,$ neural network. & & \\ \hline
\end{tabular}

\vspace{0.25 cm}

\begin{minipage}[b]{0.12\textwidth}
\centering
    \begin{tikzpicture}
\coordinate [label=., align=right]  (A) at (0,0)  ;
\coordinate [label=., align=right]  (B) at (0,-2)  ;
\draw [->,very thick, 
]($ (A)!1.5mm!(0,0) $) -- node [left, pos=0.5, align=center] {\textit{Method} \\ \textit{Complexity}}(B) ;
\end{tikzpicture}
\end{minipage}
\begin{minipage}[b]{0.6\textwidth}
    \begin{tabular}{ccccc}
Method   &  Capacity (\# Training Pts) & Multivariate Response & Functional Features & Image Features \\ \hline
\nnkcde   & Up to $\sim 10^5$ & \checkmark             &                       &                   \\
(f)\rfcde & Up to $\sim 10^6$& \checkmark      & \checkmark            &                  \\
\flexcode & Up to $\sim 10^6$&    & \checkmark            &                  \\
\deepcde  & Up to $\sim 10^8$ &    & \checkmark                      & \checkmark       
\end{tabular}
\end{minipage}
\label{tab:method-strenghts}
\end{table}

In addition to the CDE methods above, we provide the package \cdetools, which can be used for tuning and assessing the performance of CDE models on held-out validation data. 
CDE method assessment is challenging per se because we never observe the true conditional probability density, merely samples (observations) from  it. 
Furthermore, whereas loss functions such as the root-mean-squared error (RMSE) are typically used in regression problems, they are not appropriate for the task of uncertainty quantification of estimated probability densities.
The \cdetools\ package provides two types of functions for method assessment. 
First, it provides functions for computing a so-called CDE loss function (defined by Equation~\ref{eq:cde-loss} in Section \ref{sec:cde_loss}) for tuning and assessing the quality of individual PDFs. 
Second, it provides diagnostic functions that probe the population-level performance of the PDFs. 
More specifically, we have included functions for computing the Probability Integral Transform (PIT) and Highest Posterior Density (HPD); 
these metrics check how well the final density estimates on average fit the data in  the tail and highest-density regions, respectively (see Section \ref{sec:hpd}, and Figure 2 for a visual sketch).

\textbf{Organization of the paper.}
In Section \ref{sec:cde_methods}, we introduce tools for conditional density estimation (\nnkcde, \rfcde/\frfcde, \flexcode, \deepcde). 
In Section \ref{sec:assessment}, we describe tools for model selection and diagnostics.
Then, in  Section \ref{sec:examples}, we illustrate our CDE and method assessment tools for three different applications: 
photo-$z$ estimation, likelihood-free cosmological inference and spec-$z$ estimation. 
\python\ and \rlang\ code usage examples can be found in Appendix \ref{sec:code_examples}.

\textbf{Notation.} 
We denote the true (unknown) conditional density by $p(\y|\x)$, whereas an estimate of the density is
denoted by $\hat p(\y|\x)$.
We represent the CDE loss function that measures the discrepancy between the conditional density $p$ and its estimate $\hat p$ by $L(p,\hat p)$. 
Typically this loss cannot be computed directly because it depends on unknown quantities; 
an estimate of the  loss is denoted by $\hat L(p,\hat p)$.
As before, we continue to use bold-faced letters to denote vectors.

\section{Overview of Conditional Density Estimation Tools}
\label{sec:cde_methods}

We start by briefly describing the conditional density estimators in Table \ref{tab:method-strenghts}. 
Unless otherwise stated, we choose the tuning or hyper-parameters by minimizing the CDE empirical loss in Equation~\ref{eq:estimated-cde-loss} using cross-validation.

\subsection{\nnkcde}
\label{sec:NNKCDE}

Nearest-Neighbors Kernel CDE (NNKCDE; \citealt{izbicki2017photo}, \citealt{freeman2017unified}) is a simple  and easily interpretable CDE method. 
It computes a kernel density estimate of $\y$ using the $k$ nearest neighbors of the evaluation point $\mathbf{x}$.
The model has only two tuning parameters:
the number of nearest neighbors $k$ and the  bandwidth $h$ of the smoothing kernel in $\y$-space.
Both tuning parameters are chosen in a principled way by minimizing the CDE loss on validation data.
  
More specifically, the kernel density estimate of $\y$ given $\x$ is defined as
\begin{equation}
\widehat{p}(\mathbf{y} \gvn \mathbf{x}) = \frac{1}{k} \sum_{i=1}^k K_h\left[\rho(\mathbf{y},\mathbf{y}_{s_i(\mathbf{x})})\right],
\end{equation}
where $K_h$ is a normalized kernel (e.g., a Gaussian function) with bandwidth $h$, $\rho$ is a distance metric, and $s_i(\mathbf{x})$ is the index of the $i^{\rm th}$ nearest neighbor of $\mathbf{x}$. 
It is essentially a smoother version of the histogram estimator proposed by \citet{cunha2009estimating} in that it approximates the density with a smooth continuous function rather than by binning.

We provide the \nnkcde\footnote{\url{https://github.com/tpospisi/nnkcde}} software in both \python\ and \rlang\ \citep{NNKCDE_code}, with examples in Section~\ref{app:nnkcde}.
\nnkcde\ can accommodate any number of training examples. 
However, selecting tuning parameters scales quadratically in $k$, the number of nearest neighbors, which prohibits using $k \gtrsim 10^{3}$.
Our implementation (\citealt[Appendix D]{izbicki2019abc}) is computationally more efficient than standard nearest-neighbor kernel smoothers. 
For instance, we are able to efficiently evaluate the loss function in Equation~\ref{eq:estimated-cde-loss} on large validation samples by expressing the first integral in terms of convolutions of the kernel function; 
these have a fast-to-compute closed solution for a Gaussian kernel.

\subsection{\rfcde}
\label{sec:RFCDE}

Random forests (RFs, \citealt{breiman2001random}) is one of the best off-the-shelf solutions for regression and classification problems. 
It builds a large collection of decorrelated trees, where each tree is a data-based partition of the feature space.
The trees are then averaged to yield a prediction.
RFCDE, introduced by \cite{pospisil2018rfcde}, is an extension of random forests to conditional density estimation and multivariate responses. 
Like \nnkcde, it computes a kernel density estimate of $\y$ but with nearest neighbor weightings defined by the location of the evaluation point $\x$ relative to the leaves in the random forest.
RFCDE inherits the advantages of random forests in that it can handle mixed-typed data. 
It also does not require the user to specify distances or similarities between data points, and it has good performance while remaining relatively interpretable.

The main departure from other random forest algorithms is our criterion for feature space partitioning decisions. 
In regression contexts, the splitting variable and split point are typically chosen so as to minimize the mean-squared-error loss. 
In classification contexts, the splits are typically chosen so as to minimize a classification error. 
Existing random forest density estimation methods such as quantile regression forests by \cite{meinshausen2006quantile} and the TPZ algorithm by \cite{carrasco2013tpz} use the same tree structure as regression and classification random forests, respectively. 
\rfcde, however, builds trees that minimize the CDE loss (see Equation~\ref{eq:estimated-cde-loss}), allowing the forest to adapt to structures in the conditional density; hence overcoming some of the limitations of the usual regression approach for data with heteroskedasticity and multimodality. 
In addition, \rfcde\ does not require discretizing the response as in TPZ, thereby providing more accurate results at a lower cost for continuous responses, especially in the case of multivariate continuous responses where binning is problematic. 
See \cite{pospisil2018rfcde} for further examples and comparisons.

Another unique feature of \rfcde\ is that it can handle multivariate responses with joint densities by applying a weighted kernel smoother to $\y$. 
This added feature enables analysis of complex conditional distributions that describe relationships between multiple responses and features, or equivalently between multiple parameters and observables in an LFI setting. 
Like quantile regression forests, the \rfcde\ algorithm takes advantage of the fact that random forests can be viewed as a form of adaptive nearest-neighbor method with the aggregated tree structures determining a weighting scheme. 
This weighting scheme can then be used to estimate the conditional density $\pr{\mathbf{y} \gvn \mathbf{x}}$, as well as the conditional mean and quantiles, as in quantile regression forests (but for CDE-optimized trees). 
As mentioned above, \rfcde\ computes the latter density by a weighted kernel density estimate (KDE) in $\y$ using training points near the evaluation point $\mathbf{x}$.  
These distances are effectively defined by how often a training point $\mathbf{x}_i$ belongs to the same leaf node as $\mathbf{x}$ in the forest (see  Equation 1 in \citealt{pospisil2018rfcde} for details).

Despite the increased complexity of our CDE trees, \rfcde\ still scales to large data sets because of an efficient computation of splits via orthogonal series.
Moreover, \rfcde\ extends the density estimates on new $\mathbf{x}$ to the multivariate case through the use of multivariate kernel density estimators \citep{epanechnikov1969non}. 
In both the univariate and multivariate cases, bandwidth selection can be handled by either plug-in estimators or by tuning using a density estimation loss.   

For ease of use in the statistics and astronomy communities, we provide \rfcde\footnote{\url{https://github.com/tpospisi/RFCDE}} in both \python\ and \rlang, which call a common \code{C++} implementation of the core training functions that can easily be wrapped in other languages \citep{RFCDE_code}.

\textbf{Remark:} 
Estimating a CDE loss is an inherently harder task than calculating the mean squared error (MSE) in regression. 
As a consequence, \rfcde\ might not provide meaningful tree splits in settings with a large number of noisy features. 
In such settings, one may benefit from combining a regular random forest tree structure (optimized for regression) with a weighted kernel density estimator (for the density calculation).  
See Section \ref{sec:functional} for an application to functional data along with software implementation.\footnote{Available at \url{https://github.com/Mr8ND/cdetools_applications/spec_z/}} 

\newpage

\subsubsection{\frfcde}
\label{sec:frfcde}

In addition the \rfcde\ package includes \frfcde, a variant of \rfcde\ \citep{pospisil2019}, that can accommodate functional features $\x$ by partitioning in the  continuous domain of such features.
The spectral energy distribution (SED) of a galaxy is its energy as a function of continuous wavelength $\lambda$ 
of light; hence it can be viewed as functional data. 
Another example of functional data is the shear correlation function of weak lensing, which measures the mean product of the shear at two points as a function of a continuous distance $r$ between those points. Similarly, any function of continuous time is an example of functional data.
Treating functional features (like spectra, correlation functions or images) as unordered multivariate vectors on a grid suffers from a curse of dimensionality.
As the resolution of the grid becomes finer the dimensionality of the data increases but little additional information is added, due to high correlation between nearby grid points. 
\frfcde\ adapts to this setting by partitioning the  domain of each functional feature (or curve) into intervals, and passing the mean values of the function in each interval as input to \rfcde.  
Feature selection is then effectively done over regions of the domain rather than over single points. 
More specifically, the partitioning in \frfcde\ is governed by the 
parameter
$\mu$ of a Poisson process, with each functional feature entering as a high-dimensional vector $\mathbf{x} = (x_1,\ldots,x_d)$. 
Starting with the first element of the vector, we group the first
\(\operatorname{Poisson}(\mu)\) elements together.
We then repeat the procedure sequentially until we have assigned all $d$ elements into a group; 
this effectively partitions the function domain into disjoint intervals $\left\{(l_{i}, h_{i})\right\}$. 
The function mean values 
or smoothed brightness measurements $\widetilde{x}_i \equiv \int_{l_{i}}^{h_{i}} f(\lambda) d\lambda$ of each interval are finally treated as new inputs to a standard  (vectorial) \rfcde\ tree.
The splitting of the smoothed predictors $\widetilde{x}_i$ is done independently for each tree in the forest.  
Other steps of \frfcde, such as the computation of variable importance, also proceed as in (vectorial) \rfcde\ but with the averaged values of a region as inputs. 
As a result, \frfcde\ has the capability of identifying the functional inputs and the regions in the input domain that are important for estimating the response $\y$.
Figure \ref{fig:rfcde_viz} shows schematically the differences and similarities in construction between standard \rfcde\ and its \frfcde\ variant.

The \frfcde\ method is as scalable as standard random forests, accommodating training sets on the order of $10^6$. 
As the examples in Section~\ref{sec:functional} show, we can obtain substantial gains with a functional approach, both in terms of statistical performance (that is, CDE loss) and computational time. 
In addition, the change in code syntax is minimal, as one only needs to pass the  
Poisson $\mu$ parameter as the \texttt{flambda} argument during the forest initialization. 
Examples in \python\ and \rlang\ are provided in Appendix~\ref{app:frfcde}.

\begin{figure}
    \centering
    \includegraphics[width=0.5\textwidth]{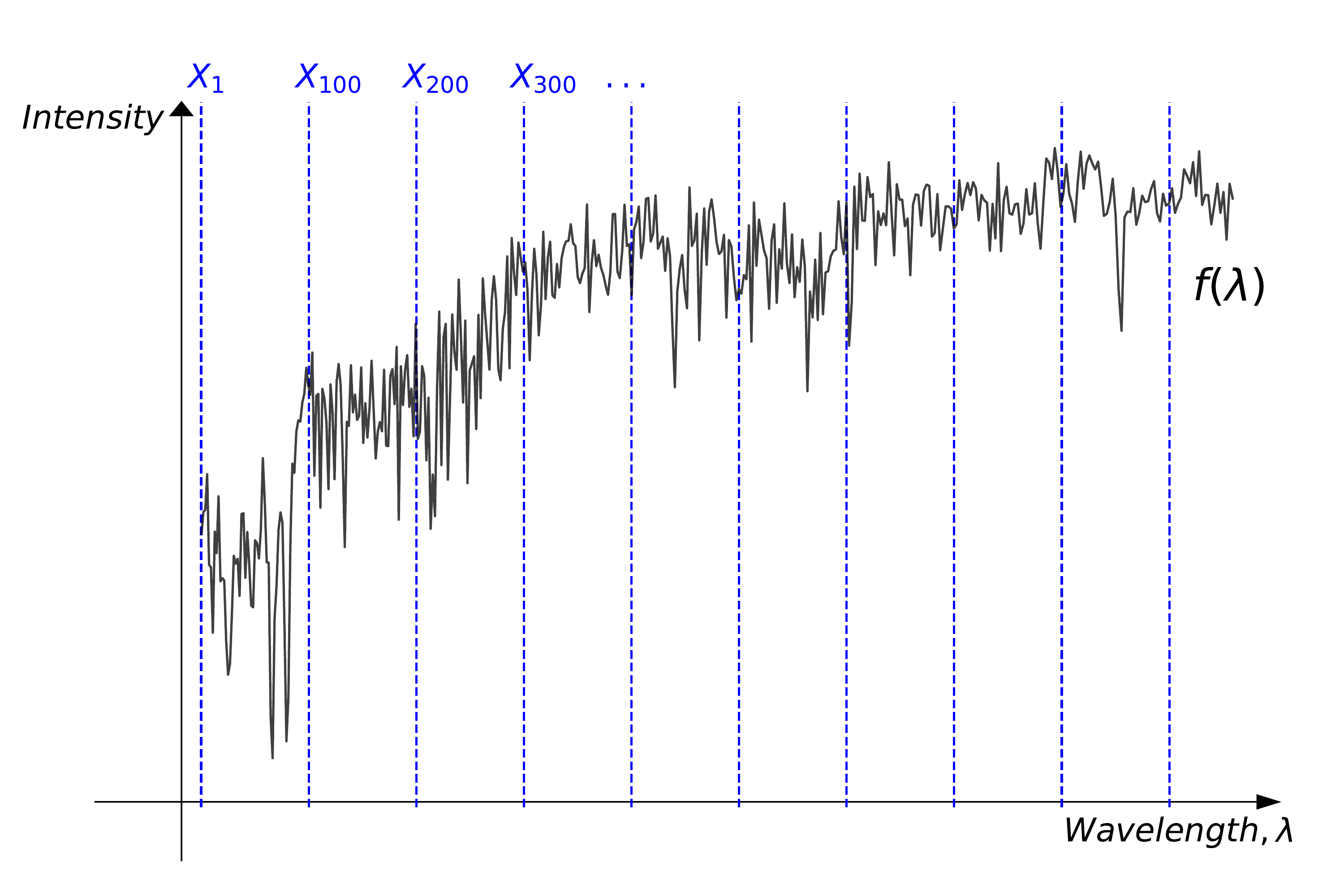}%
    \includegraphics[width=0.5\textwidth]{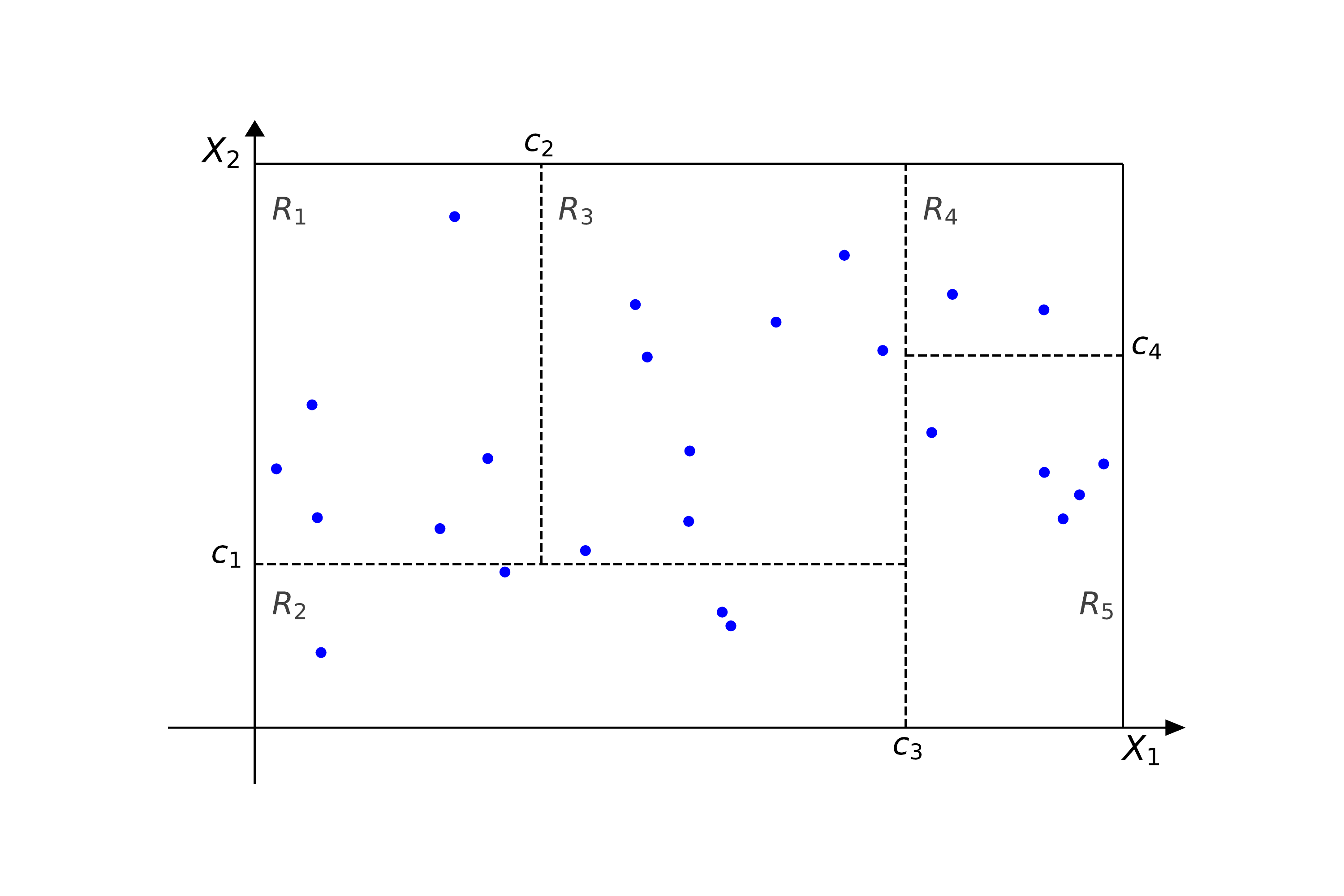}\\
    \includegraphics[width=0.5\textwidth]{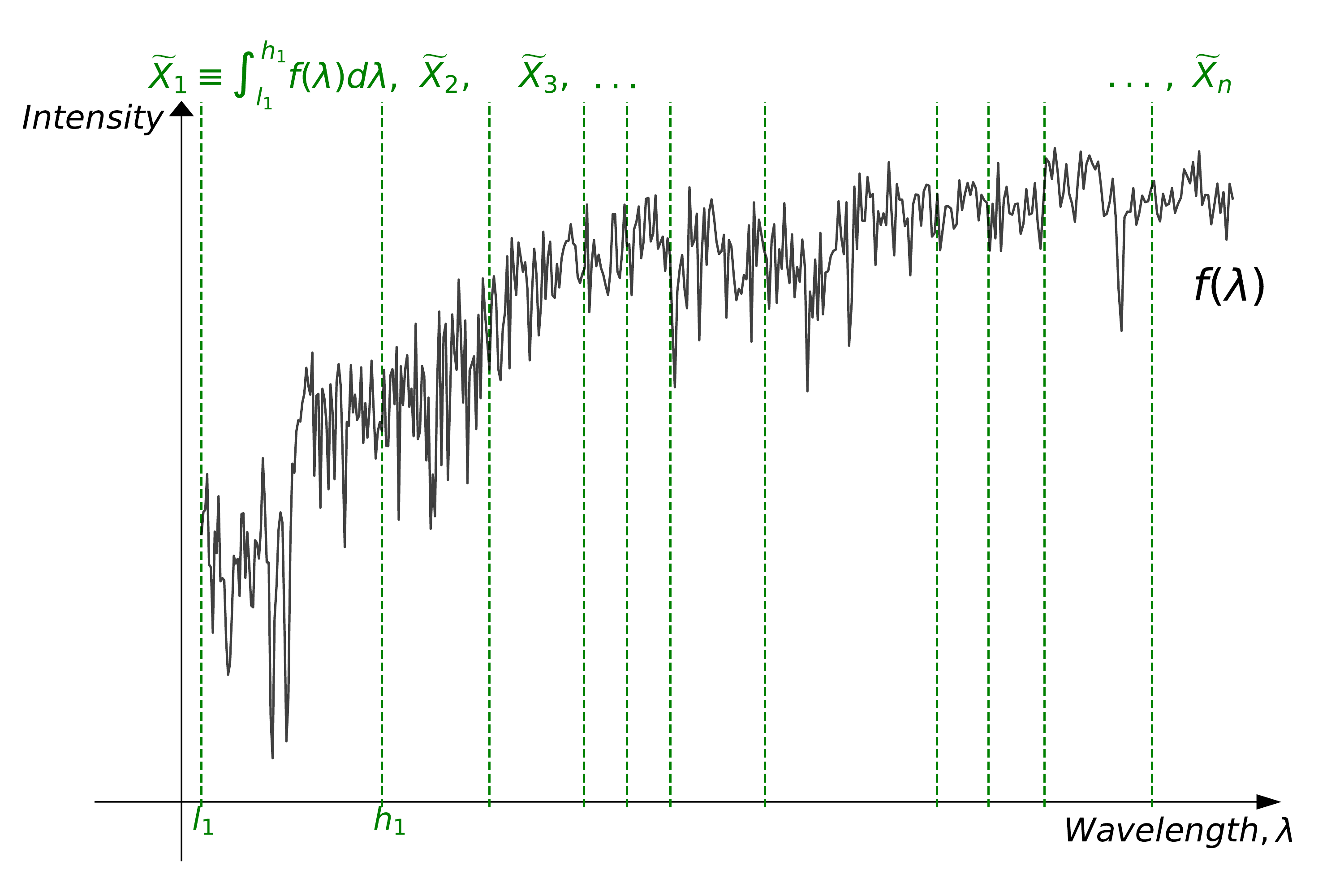}%
    \includegraphics[width=0.5\textwidth]{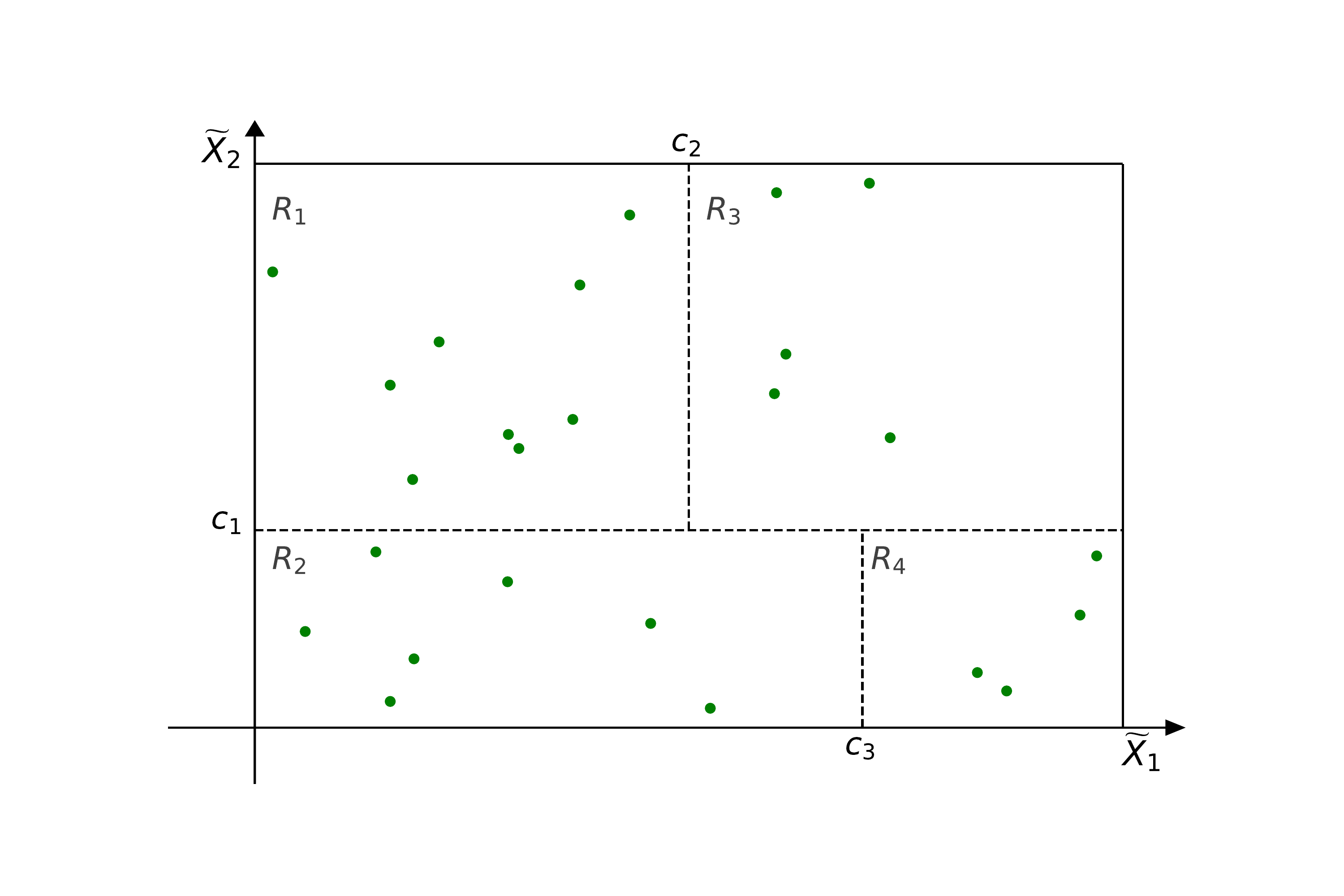}
    \caption{ 
    A schematic diagram of \rfcde\ (top row) and \frfcde\ (bottom row) applied to a galaxy spectrum from \cite{astroML2012}. 
    \emph{Top row}: 
    \rfcde\ treats the intensity $x_i$ at each recorded wavelength $\lambda_i$ of the spectrum as a feature or ``input'' to the random forests algorithm --- the blue vertical dashed lines indicate every 100$^{\rm th}$ recorded wavelength. 
    \rfcde\ then builds an ensemble of CDE trees, where each tree partitions the feature space according to the CDE loss, as illustrated in the top right figure for features $x_1$ and $x_2$. 
    \emph{Bottom row}: 
    \frfcde\ instead groups nearby measurements together where the group divisions are defined by a Poisson process with parameter $\mu$ (vertical green dashed lines, left figure). 
    The new smoothed features $\widetilde{x_i}$ are computed by integrating the intensity over the grouped wavelengths. 
    A forest of CDE trees is then built using the same construction as in \rfcde\ but with the smoothed features as inputs (bottom right figure).
    }\label{fig:rfcde_viz}
\end{figure}

\subsection{\flexcode}
\label{sec:FlexCode}

Introduced by \cite{izbicki2017converting}, \flexcode\footnote{ \url{https://github.com/tpospisi/FlexCode} (\python) and \url{https://github.com/rizbicki/FlexCoDE} (\rlang)} 
is a CDE method that uses a basis expansion for the univariate response $y$ and poses CDE as a series of univariate regression problems. 
The main advantage of this method is its flexibility as any regression method can be applied towards CDE, enabling us to tailor our choice of regression  method to the intrinsic structure and type of data at hand.

More precisely, let $\left\{\phi_{j} (y) \right\}_j$ be an orthonormal basis like a Fourier or wavelet basis for
functions of $y \in \mathbb{R}$.
The key idea of \flexcode\ is to express the unknown density $\pr{y \gvn \mathbf{x}}$ as a basis expansion
\begin{equation}
\label{eq:flexcode}
  p(y \gvn \mathbf{x}) = \sum_{j} \beta_{j}(\mathbf{x}) \phi_{j}(y) .
\end{equation}
By the orthogonality property of the basis, the (unknown) expansion coefficients $\{\beta_{j}(\mathbf{x})\}_j$ are then just orthogonal projections of $p(y|\x)$ onto the basis vectors. 
We can estimate these coefficients using a training set of $(\mathbf{x},y)$ data by {\em regressing} the transformed response variables $\phi_j(y)$ on predictors $\mathbf{x}$ for every basis function $j$ (see \cite{izbicki2017converting} Equation 2.2, for details). 
The number of basis function $n_{basis}$ is chosen by minimizing a CDE loss function on validation data.
The estimated density, $\sum_{j=1}^{n_{basis}} \widehat{\beta}_j(\mathbf{x}) \phi_j(y)$, may contain small spurious bumps induced by the Fourier approximation and may not integrate to one.
We remove such artifacts as described in \citet{izbicki2016nonparametric} by applying a thresholding parameter $\delta$ chosen via cross-validation.
\flexcode\ turns a challenging density estimation problem into a simpler regression problem, where we can choose any regression method that fits the problem at hand. 

To provide a concrete example, for high-dimensional $\x$ (such as galaxy spectra) we can use methods such as sparse or spectral series (``manifold'') regression methods \citep{tibshirani1996regression,ravikumar2009sparse,lee2016spectral}; 
see Section~\ref{sec:functional} for an example with \flexcode\code{-Series}. 
For multi-probe studies with mixed data types, we can use random forests regression \citep{breiman2001random}. 
On the other hand, large-scale photometric galaxy surveys such as \lsst\ 
require methods that can work with data from millions, if not billions, of galaxies. 
\citet{desc_photoz} present the results of an initial study of the \lsst\ Dark Energy Science Collaboration (\textsc{LSST-DESC}).
Their initial data challenge (``\Pz\ DC 1'') compares the CDEs of a dozen \pz\ codes run on simulations of \lsst\ galaxy photometry catalogs in the presence of complete, correct, and representative training data.
\flexzboost, a version of \flexcode\ based on the scalable gradient boosting regression technique by \cite{chen2016xgboost}, was entered into the data challenge because of the method's ability to scale to massive data. 
In the DC1 analysis, \flexzboost\ was among the the strongest performing codes according to established performance metrics of such PDFs and was 
one of only two codes to beat the experimental control under a more discriminating metric, the CDE loss.

For massive surveys such as \lsst, \flexcode\ also has another advantage compared to other CDE methods, namely its compact, lossless storage format.
\citet{juric_data_2017} establishes that \lsst\ has allocated $\sim$100 floating point numbers to quantify the redshift of each galaxy.
As is shown in \citet{desc_photoz}, the myriad methods for deriving \pz\ PDFs yield radically different results, motivating a desire to store the results of more than one algorithm in the absence of an obvious best choice.
For the \pz\ PDFs of most codes, one may need to seek a clever storage parameterization to meet \lsst's constraints \citep{carrasco_kind_sparse_2014, malz_approximating_2018}, but \flexcode\ is virtually immune to this restriction.
Since \flexcode\ relies on a basis expansion, one only needs to store $n_{basis}$ coefficients per target density for a lossless compression of
the estimated PDF with no need for binning. 
Indeed, for DC1, we can with \flexzboost\ reconstruct our estimate $\widehat{p}(z \gvn \mathbf{x})$ at any resolution from estimates of the first 35 coefficients in a Fourier basis expansion. 
In other words, \flexzboost\ enables the creation and storage of high-resolution \pz\ catalogs for several billion galaxies at no added cost.

Our public implementation of \flexcode
--- available in both \python\ \citep{Flexcode_code} and \rlang\ \citep{Flexcode_codeR},  respectively --- cross-validates over regression tuning parameters (such as the number of nearest neighbors $k$ in \flexcode\code{-kNN}) and computes the \flexcode\ coefficients in parallel for further time savings.  
The computational complexity of \flexcode\ will be the same as $n_{basis}$ parallel individual regressions.  
In particular, the scalability of \flexcode\ is determined by the underlying regression method. 
A scalable method like XGBoost leads to scalable \flexcode\ fitting. 
In the \python\ version of the code, the user can choose between the following regression methods: XGBoost for \flexzboost\, but also nearest neighbors, LASSO \citep{tibshirani1996regression}, and random forests regression \citep{breiman2001random}.  
In the \rlang\ version, the following choices are available: 
XGBoost, nearest neighbors, LASSO, random forests, Nadaraya-Watson kernel smoothing \citep{nadaraya1964estimating,watson1964smooth}, sparse additive models \citep{ravikumar2009sparse}, and spectral series (``manifold'') regression \citep{lee2016spectral}. 
In both implementations, the user may also use a custom regression method; 
we illustrate how this can be done with a vignette in both packages. 
For the \python\ implementation, the user can add any custom regression method following the \code{sklearn} API, i.e., with \texttt{fit} and \texttt{predict} methods.
Examples in both languages are presented in Appendix~\ref{app:flexcode}.

\subsection{\deepcde}
\label{sec:DeepCDE}

Recently, neural networks have reemerged as a powerful tool for prediction settings where large amounts of representative training data are available; 
see \cite{lecun2015deep} and \cite{goodfellow2016deep} for a full review.
Neural networks for CDE, such as Mixture Density Networks (MDNs; \citealt{bishop1994mixture}) and variational methods \citep{Tang2013varationalCDE, Sohn2015generativeCDE}, usually assume a Gaussian or some other parametric form of the conditional density. 
MDNs have lately also been used for photometric redshift estimation \citep{d2018photometric, pasquet2019photometric} and for direct estimation of likelihoods and posteriors in cosmological parameter inference (see \citealt{alsing2019fast} and references within). 

\deepcde\footnote{\url{https://github.com/tpospisi/DeepCDE}} \citep{DeepCDE_code} takes a different, {\em fully nonparametric} approach to CDE. 
It combines the advantages of basis expansions with the flexibility of neural network architectures, allowing for data types like image features and time-series data.
\deepcde\ is based on the orthogonal series representation in \flexcode, given in Equation \ref{eq:flexcode}, but rather than relying on regression methods to estimate the expansion coefficients in Equation~\ref{eq:flexcode}, \deepcde\ computes the coefficients $\left\{\beta_{i}(\mathbf{x})\right\}_{i=1}^{B}$ {\em jointly} with a  neural network that minimizes the CDE loss in Equation~\ref{eq:cde-loss}. 
Indeed, one can show that for an orthogonal basis, the problem of minimizing this CDE loss is (asymptotically) equivalent to finding the best basis coefficients in \flexcode\ under mean squared error loss for the individual regressions; 
see Appendix~\ref{app:cdeloss_equiv} for a proof.
The value of this result is that \deepcde\ with a CDE loss directly connects prediction with uncertainty quantification, implying that one can leverage the state-of-the-art deep architectures for an application at hand toward uncertainty quantification for the same prediction setting.

From a neural network architecture perspective, \deepcde\ only adds a linear output layer of coefficients for a series expansion of the density according to 
\begin{equation}
 \label{eq::expansion_DeepCDE}
   \hat{p}(y \gvn \mathbf{x}) = \sum_{j=1}^{B} \widehat{\beta}_{j}(\mathbf{x}) \phi_{j}(y) ,
\end{equation}
where $\{\phi_{j}(y)\}_{j=1}^B$ is an orthogonal basis for functions of $y \in \mathbb{R}$. 
Like \flexcode, we normalize and remove spurious bumps from the final density estimates according to the procedure in Section 2.2 of \cite{izbicki2016nonparametric}.

The greatest benefit of \deepcde\ is that it can be implemented with both convolutional and recurrent neural network architectures, extending to both image and sequential data.
For most deep architectures, adding a linear layer represents a small modification, and a negligible increase in the number of parameters.  
For instance, with the AlexNet architecture \citep{Krizhevsky2012Alexnet}, 
a widely used, relatively shallow convolutional neural network, adding a final layer with 30 coefficients for a cosine basis only adds $\sim 120,000$ extra parameters. 
This represents a $0.1\%$ increase over the $62$ million already existing parameters, and hence a negligible increase in training and prediction time. 
Moreover, the CDE loss for \deepcde\ is especially easy to evaluate; see Appendix \ref{app:cdeloss_equiv} for details.

We provide both \code{TensorFlow} and \code{Pytorch} implementations of \deepcde\ \citep{DeepCDE_code}.
We also include examples that shows how one can easily build \deepcde\ on top of AlexNet \citep{Krizhevsky2012Alexnet}; 
in this case, for the task of estimating the probability distribution $\pr{y \gvn \mathbf{x}}$ of the correct orientation $y$ of a color image $\mathbf{x}$. 
Note that \cite{Fischer2015ImageOE} use AlexNet for the corresponding regression task of predicting the orientation $y$ for a color image $\mathbf{x}$ but without quantifying the uncertainty in the predictions.

\section{How to Assess Method Performance}
\label{sec:assessment}

After fitting CDEs, it is important to assess the quality of our models of uncertainty. 
For instance, after computing \pz\ PDF estimates $\widehat{p}(z \gvn \mathbf{x})$ for some galaxies, one may ask whether these estimates accurately quantify the true uncertainty distributions $\pr{z\gvn \mathbf{x}}$. 
Similarly, in the LFI task, a key question is whether an estimate of the posterior distribution, $\widehat{p}(\boldsymbol{\theta} \gvn \xobs)$ of the cosmological parameters is close enough to the true posterior $\pr{\boldsymbol{\theta} \gvn \xobs}$ given the observations $\xobs$.

We present two method-assessment tools suitable to different situations, which are complementary and can be performed simultaneously, with public implementations in the \cdetools\footnote{\url{https://github.com/tpospisi/cdetools}} package in both \python\ and \rlang\ \citep{cdetools_code0}.
First, we describe a CDE loss function that directly provides relative comparisons between conditional density estimators, such as the methods presented in Section \ref{sec:cde_methods} or, equivalently, between a set of models (for the same method) with different tuning parameters.
Second, we describe visual diagnostic tools, such as Probability Integral Transforms (PIT) and Highest Probability Density (HPD) plots, that can provide insights on the overall goodness-of-fit of a given estimator to observed data.

\subsection{CDE loss}
\label{sec:cde_loss}

Here we briefly review the CDE loss from \cite{izbicki2016nonparametric} for assessing conditional density estimators and discuss it in the context of the cosmology LFI case.

The goal of a loss function is to provide relative comparisons between different estimators, so that it is easy to directly choose the best fitted model among a list of candidates. 
Given an estimate $\widehat{p}$ of $p$, we define the CDE loss by
\begin{equation}
  \label{eq:cde-loss}
  L(\widehat{p}, p) = \int \int \left[\widehat{p}(\y | \x) - p(\y | \x)\right]^{2} d\y dP(\x) ,
\end{equation}
where $P(\mathbf{x})$ is the marginal distribution of the features $\x$. 
This loss is the CDE analog to the standard mean squared error (MSE) in standard regression.
The weighting by the marginal distribution of the features emphasizes that errors in the estimation of $\y$ for unlikely features $\x$ are less important.
The CDE loss cannot be directly evaluated because it depends on the unknown true density $p(z | \x)$. However, one can estimate the loss (up to a constant determined by the true $p$) by
\begin{equation}
  \label{eq:estimated-cde-loss}
  \hat{L}(\widehat{p}, p) = \frac{1}{n} \sum_{i=1}^{n} \int \widehat{p}(\y | \mathbf{x}_{i}^{te})^{2} d\y - \frac{2}{n} \sum_{i=1}^{n} \widehat{p}(\y_{i}^{te} | \mathbf{x}_{i}^{te}) ,
\end{equation}
where $\left\{(\mathbf{x}_{i}^{te}, \y_{i}^{te})\right\}_{i=1}^{n}$ represents our validation or test data, i.e., a held-out set not used to construct $\widehat{p}$. 
In our implementation, the function \texttt{cde\_loss} returns the estimated CDE loss as well as an estimate of the standard deviation or the {\em standard error} (SE) of the estimated loss. 

{\bf CDE loss for LFI.} In LFI settings, we use a slightly different version of the CDE loss in Eq.~\ref{eq:cde-loss}. Because the goal (in ABC) is to approximate the posterior density $p(\boldsymbol{\theta} \gvn \xobs)$ at fixed $\x=\xobs$, a natural evaluation metric is the integrated squared error loss
\begin{equation}
  \label{eq:abc-loss}
  \int \left[\hat{p}(\boldsymbol{\theta} \gvn \xobs) - p(\boldsymbol{\theta} \gvn \xobs)\right]^{2} d\boldsymbol{\theta}
\end{equation}
of the conditional density \emph{at $\xobs$ only}. 
Estimating this loss can however be tricky as only a single instance of data with $x=\xobs$ is available in practice. 
Hence, \cite{izbicki2019abc} approximates Equation \ref{eq:abc-loss} by computing the empirical loss $\widehat{L}(\widehat{p}, p)$ in Eq.~\ref{eq:estimated-cde-loss} over a restricted subset of the validation data that only includes the  $\mathbf{x}_{i}^{te}$ points that fall in an $\epsilon$-neighborhood of $\xobs$, where $\epsilon$ is the tolerance of the ABC rejection algorithm. 
The detailed analysis of this approach can be found in \cite{izbicki2019abc}.

\subsection{PIT and HPD diagnostics} 
\label{sec:hpd}

The CDE loss function is a relative measure of performance that cannot address absolute goodness-of-fit.
To quantify overall goodness-of-fit, we examine how well calibrated an ensemble of conditional density estimators are on average, over validation or test data $\left\{(\mathbf{x}_{i}^{te}, \y_{i}^{te})\right\}_{i=1}^{n}$. For ease of notation, we will in this section denote $\mathbf{x}_{i}^{te}$ and  $\y_{i}^{te}$ for a generic $i$ by $\x_{val}$ and $\y_{val}$.

Given a true probability density $p(\y \gvn \x) = \gamma_{0}$ of a variable $\y$ conditioned on data $\x$, an estimated probability density $\widehat{p}(\y \gvn \x) = \gamma$ cannot be {\em well-calibrated} unless $\gamma \approx \gamma_{0}$.
Built on the same logic, the probability integral transform (PIT; \citealt{Polsterer})
\begin{equation}
PIT(\xobsph, \yobsph) = \int_{-\infty}^{\yobsph} \hat{p}(y \gvn \xobsph) dy 
\end{equation}
assesses the calibration quality of an ensemble of CDEs for scalar $y$ representing the cumulative distribution function (CDF) of $\hat{p}(y \gvn \xobsph)$ evaluated at $y = \yobsph$;  this PIT value corresponds to the shaded area in Figure \ref{fig:pit_hpd_visualization}, left. 
A statistically self-consistent population of densities has a uniform distribution of PIT values, and deviations from uniformity indicate inaccuracies of the estimated PDFs. 
Overly broad CDEs manifest as under-representation of the lowest and highest PIT values, whereas overly narrow CDEs manifest as over-representation of the lowest and highest values.

However, PIT values do not easily generalize to multiple responses.
For instance, for a bivariate response $\y = (z, \eta)$, the quantity $\int_{-\infty}^{\zobsph} \int_{-\infty}^{\etaobsph} \pr{z, \eta \gvn \xobsph} dz d\eta$ is not in general uniformly distributed \citep{genest2001multivariate}.
An alternative statistic that easily generalizes to multivariate $\y$ is the highest probability density value (HPD; \citealt[Appendix A]{izbicki2017photo}):
\begin{equation} \label{eq:HPD}
\xi(\xobsph, \Yobsph) = \int_{\textbf{y}:\hat{p}(\textbf{y} \gvn \xobsph) \ge \hat{p}(\Yobsph \gvn \xobsph)} \hat{p}(\textbf{y} \gvn \xobsph) d\textbf{y} .
\end{equation}
The HPD value is based on the definition of the {\em highest density region} (HDR, \citealt{hyndman1996computing}) of a random variable $\y$; 
that is, the  subset of the sample space of $\y$ where all points in the region have a probability above a certain value. 
The HDR of $\textbf{y} \gvn \xobsph$ can be seen as a region estimate of  $\textbf{y}$ when $\x=\xobsph$ is observed.
In words, the set $\{\textbf{y}:\hat{p}(\textbf{y} \gvn \xobsph) \ge \hat{p}(\Yobsph \gvn \xobsph)\}$ is the smallest HDR containing the point $\Yobsph$, and the HPD value 
is simply the probability 
of such a region. 
Figure \ref{fig:pit_hpd_visualization}, right, shows a schematic diagram of the HPD value (green shaded area) and HDR region (highlighted segments on the y-axis) for the estimated density $\hat{p}(\textbf{y} \gvn \xobsph)$. 
The HPD value $\xi(\xobs, \Yobsph)$ can also be viewed as a measure of how plausible $\Yobsph$ is according to $\hat p (\y | \xobsph)$ and is directly related to the Bayesian analog of p-values or {\em the e-value} \citep{pereira1999}. 
One can show \citep{harrison2015validation} that even for multivariate $\y$, the
HPD values for validation data  follow a $U(0,1)$  distribution if the CDEs are well-calibrated on the population level. Thus, these values can also be used for assessing the fit of conditional densities in the same way as PIT values.
In our implementation the functions \texttt{pit\_coverage(cde\_test, y\_grid, y\_test)} and \texttt{hpd\_coverage(cde\_test, y\_grid, y\_test)}, respectively, calculate PIT and HPD values for CDE estimates \texttt{cde\_test} using the grid \texttt{y\_grid} with observed values \texttt{y\_test}.  

The PIT and HPD are not without their limitations, however, as demonstrated in the control case of \citet{desc_photoz} and Figure~\ref{fig:setting1_diagnostics} of Section~\ref{sec:univ_response_mult_data} here.
Because the PIT and HPD values can be uniformly distributed even if $\pr{\y \gvn \mathbf{x}}$ is not well estimated, they must be used in conjunction with loss functions for method assessment.
A popular way of visualizing PIT and HPD diagnostics for the entire population is through   {\em  probability-probability plots} or P-P plots of the empirical distribution of the  (PIT or HPD) statistic  versus its distribution under the hypothesis that $\hat p(\y | \x)= p(\y | \x)$; henceforth, we will refer to the latter  Uniform(0,1) distribution as the ``theoretical'' distribution of PIT or HPD.  An ideal P-P plot has all points close to the identity line where the ``empirical'' and ``theoretical'' distributions are the same. Note that HPD P-P plots, in particular, are valuable calibration tools if our goal is to calibrate the estimated densities so that the computed  predictive regions have the right coverage.

Figure \ref{fig:teddy-pvals} (c)-(d) contains example P-P plots for the PIT and HPD values of \pz\ estimators, where the empirical (observed) distribution is close to the theoretical (ideal) distribution.

\begin{figure}
    \centering
    \includegraphics[width=0.5\textwidth]{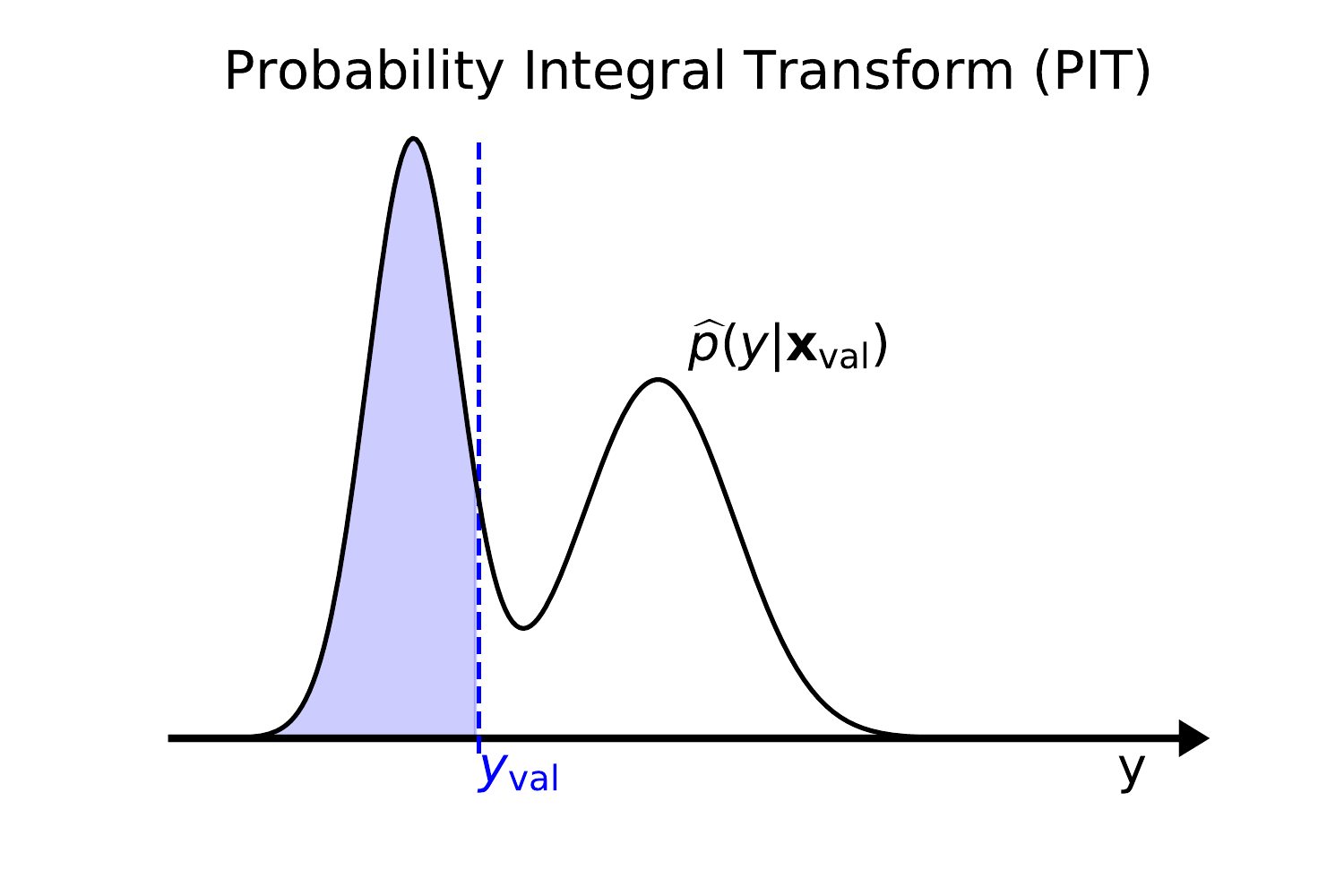}%
    \includegraphics[width=0.5\textwidth]{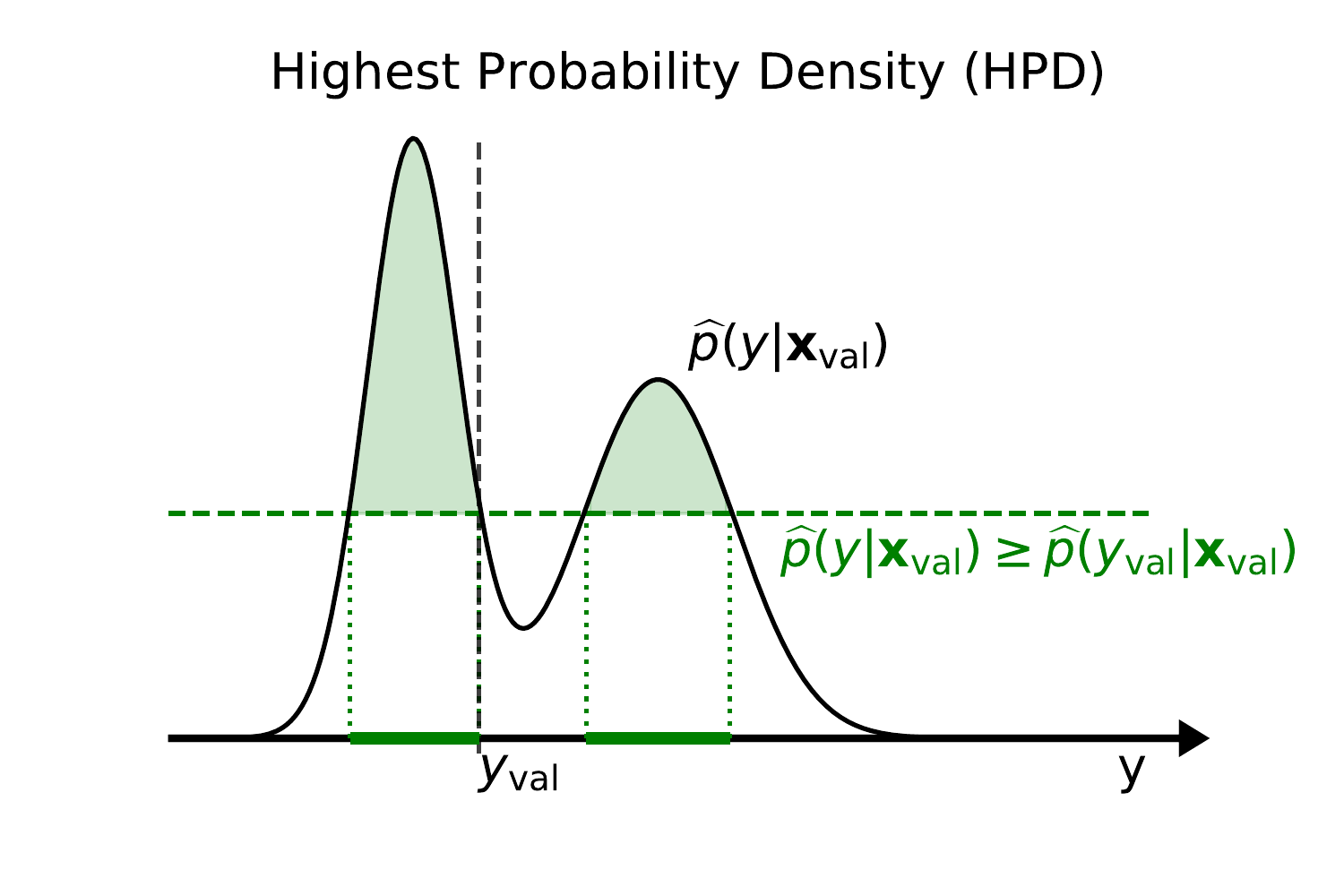}
    \caption{ 
    Schematic diagram of the construction of the Probability Integral Transform (PIT, left) and  the Highest Probability Density (HPD, right) values for  the estimated density $\widehat{p}(y|\x)$ at $\x=\xobsph$, where  $\yobsph$ is the response at $x=\xobsph$. In the  plot to the right, the highlighted segments on the $y$-axis form the so-called highest density region (HDR) of $y|\xobsph$. The PIT and HPD values correspond to the area of the tail versus highest density region, respectively, of the estimate; here indicated by the blue versus green shaded areas.
    }\label{fig:pit_hpd_visualization}
\end{figure}

\newpage

\section{Applications in Astronomy}
\label{sec:examples}

We demonstrate\footnote{Code for these examples is publicly available at \url{https://github.com/Mr8ND/cdetools_applications}.} the breadth of our CDE methods in three different astronomical use cases.

\begin{enumerate}
  \item 
  {\bf \Pz\ estimation: univariate response with multivariate input.}
  This is the standard prediction setting in which all four methods apply. 
  In Section \ref{sec:univ_response_mult_data}, we apply the methods to the \teddy\ photometric redshift data \citep{beck2017realistic}, and illustrate the need for loss functions to properly assess PDF estimates of redshift $z$ given photometric colors $\x$.
  \item
  {\bf Likelihood-free cosmological inference: multivariate response.}
  For multiple response components we want to model the often complicated dependencies between these components;
  this is in contrast to approaches which model each component separately, implicitly introducing an assumption of conditional independence. 
  In Section \ref{sec:mult_response_mult_data}, we use an example of LFI for simulated weak lensing shear data to show how \nnkcde\ and \rfcde\ can capture more challenging bivariate distributions with curved structures; 
  in this toy example $\y$ represents cosmological parameters $(\Omega_M, \sigma_8)$ in the $\Lambda$CDM-model, and $\x$ represents (coarsely binned) weak lensing shear correlation functions.
  \item 
  {\bf Spec-$z$ estimation: functional input.}
  Standard prediction methods, such as random forests, do not typically fare well with functional features, simply treating them as unordered vectorial data and ignoring the functional structure.
  However, often there are substantial benefits to explicitly taking advantage of that structure as in \frfcde.  
  In Section \ref{sec:functional}, we compare the performance of a vectorial implementation of \rfcde\ with \frfcde\ and \flexcode\code{-Series} for a spectroscopic sample from the Sloan Digital Sky Survey (\sdss; \citealt{alam2015eleventh}). The input $\x$ is here a high-resolution spectrum of a galaxy, and the response is the galaxy's redshift  $z$.
\end{enumerate}

Throughout this section we will report CDE loss mean and standard error for each method, using $95\%$ Gaussian confidence intervals for performance comparison. 
Since the CDE loss is an empirical mean and the test size is reasonably large in all the examples, the validity of this approximation is guaranteed by the central limit theorem.

\subsection{\Pz\ estimation: Univariate response with multivariate input}
\label{sec:univ_response_mult_data}

Here we estimate \pz\ posterior PDFs $\pr{z \gvn \mathbf{x}}$ using representative training and test data (Samples A and B, respectively) from the \teddy\ catalog by \cite{beck2017realistic}.\footnote{Data available at \url{https://github.com/COINtoolbox/photoz_catalogues}} These data include 74309 training and 74557 test observations.
Each observation $\mathbf{x}$ has five features: 
the magnitude of the $r$-band and the pairwise differences or ``colors'' $u-g$, $g-r$, $r-i$, and $i-z$. 

Among the chief sources of uncertainty affecting photo-$z$s estimated by ML techniques are the incompleteness and nonrepresentativity of training sets, defined by the mismatch in the distributions of training and test data in $z$ and $\mathbf{x}$, which may be extreme to the point of not guaranteeing mutual coverage.
Realistically modeling incompleteness is highly challenging, requiring both simulations of SEDs and of the observational conditions of a given survey, which is outside the scope of this work.
Accounting for redshift incompletenss is not a lost cause and may be accomplished by extrapolating outside of the training range by abandoning standard instance-based ML algorithms (see e.g., \cite{leistedt2016_gp}). 
Certain types of selection bias, known as covariate shift, can also be corrected by importance weights in the CDE loss \citep{izbicki2017photo, freeman2017unified}; 
see Appendix~\ref{sec:cov_shift_photoz} for details and code.
However, for simplicity, in this paper we consider only representative training sets with no disparity in color-space coverage, putting this demonstration on equal footing with all previous comparisons of photo-$z$ PDF methods.

We fit \nnkcde, \rfcde, \deepcde, and {\flexzboost}  
to the data. 
In addition we fit an \rfcde\ model, ``\rfcde\code{-Limited},'' restricted to the first three of the five features, as a toy model which fails to extract some information from the features, allowing us to showcase the difference between PIT or HPD diagnostics and the CDE loss function. 
For \deepcde\ we use a three-layer deep neural network with linear layers and reLu activations with $25$ neurons per layers, trained for $10,000$ epochs with Adam \citep{Kingma2014AdamAM}. 
As our goal is to showcase its applicability we do not optimize the neural architecture (e.g., number of layers, number of neurons per layer, activation functions) nor the learning parameters (i.e., number of epochs, learning rate, momentum).
To illustrate our validation methods, we also include the marginal distribution $\widehat{p}(z) = \frac{1}{n}\sum_{i=1}^n \widehat{p}(z | \mathbf{x}_i)$ as an estimate of individual \pz\ distributions $p(z | \mathbf{x})$.
This estimate will be the same regardless of $\mathbf{x}$.

\begin{figure}[htbp]
\centering
\includegraphics[width=1\linewidth]{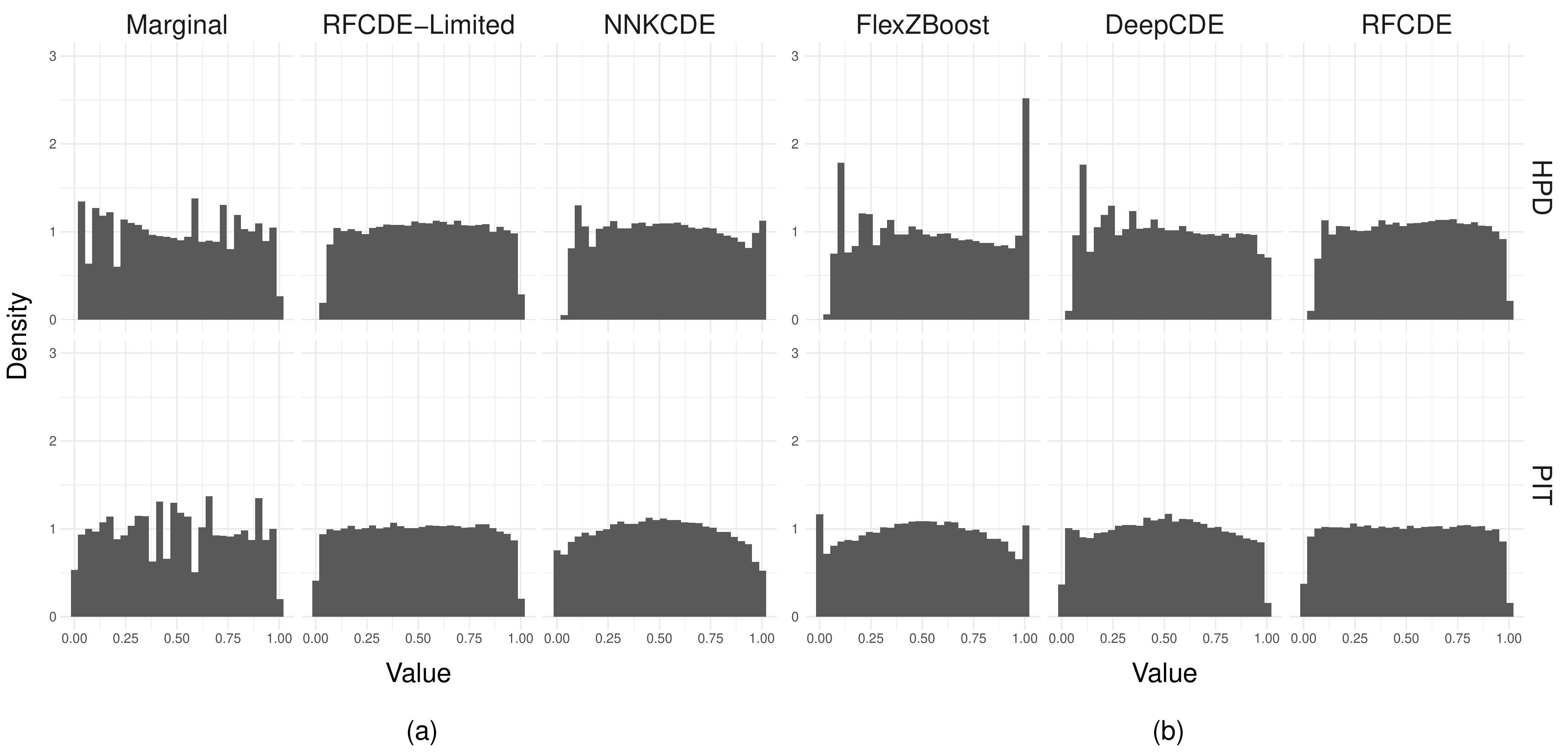}

\vspace{0.75cm}

\includegraphics[width=1\linewidth]{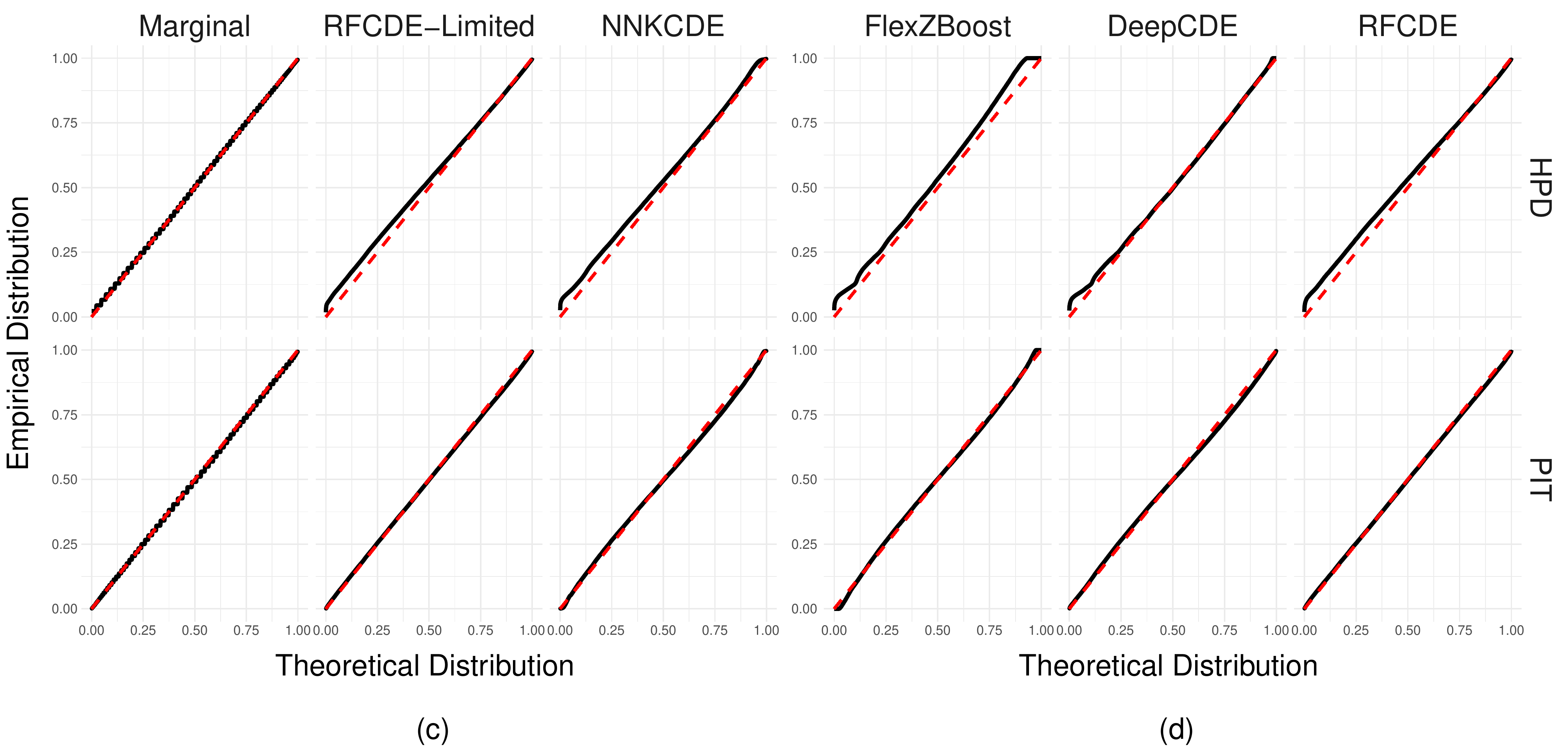}
\caption{
(a)-(b) Histograms of PIT and HPD values for \teddy\ \pz\ data. 
Both versions of \rfcde\ as well as the marginal distribution have a uniform distribution indicating  that the conditional density estimates are well-calibrated on average. 
\flexzboost\ exhibits an overrepresentation of 0 and 1 values which indicates overly narrow CDEs.
(c)-(d) Probability-Probability (P-P) plots of HPD and PIT values. Both sets of values are computed over   data in the test set  and their empirical distribution is plotted against the uniform $U(0,1)$ distribution; i.e., the ``theoretical distribution'' of the HPD and PIT values when $\hat{p}(z|\x)=p(z|\x)$. If the estimated CDEs are well calibrated, the  empirical and theoretical distributions should coincide and all points should be close to the identity line.
As in the top panel, the P-P plots indicate a good fit for all methods including the clearly misspecified ``Marginal'' model. 
 \label{fig:setting1_diagnostics}}\label{fig:teddy-pvals}
\end{figure}

\begin{figure}[htbp]
\begin{minipage}[t]{0.5\textwidth}
\vspace{0pt}
\centering
\includegraphics[width=1\textwidth]{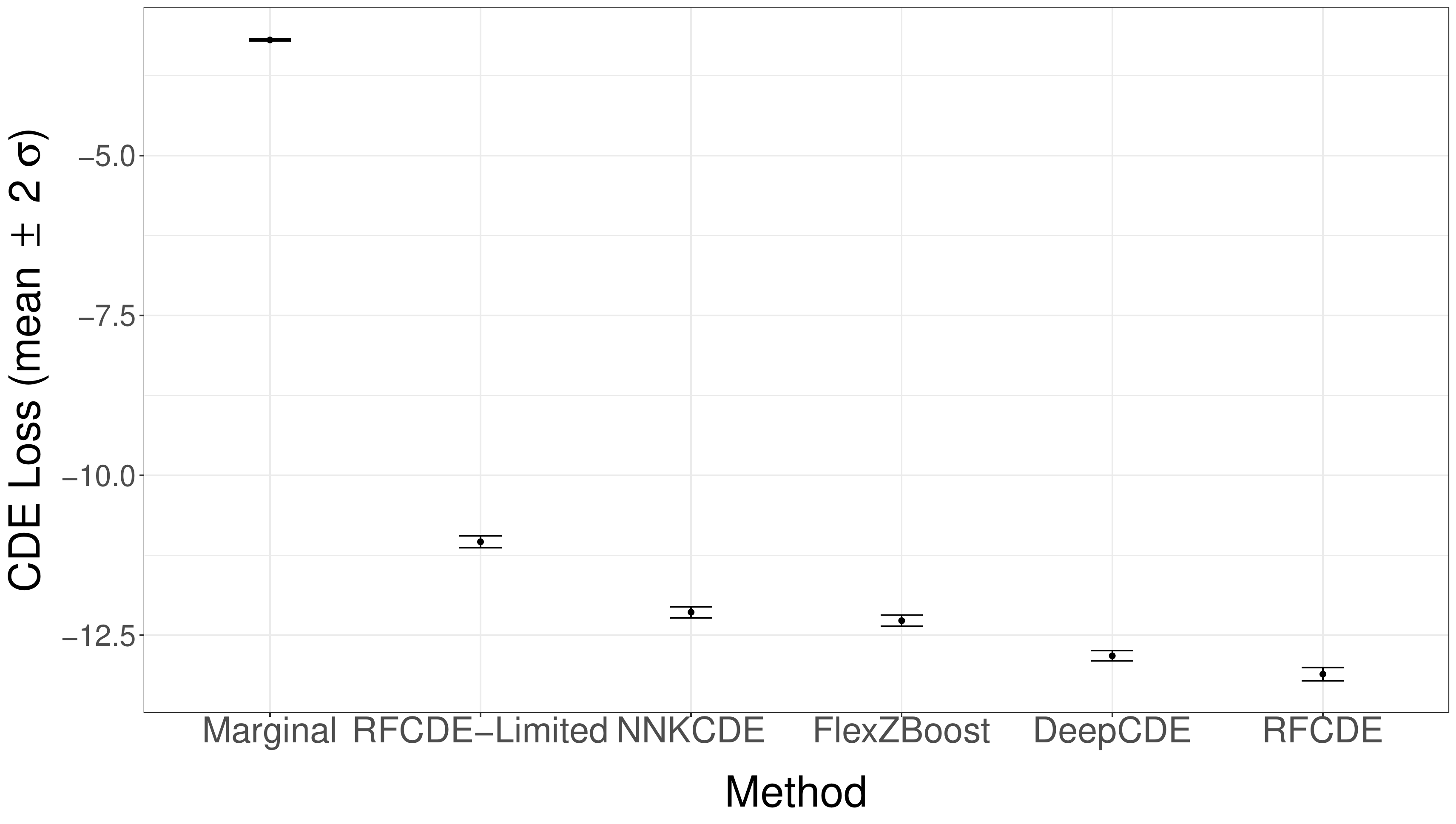}
\end{minipage}%
\begin{minipage}[t]{0.4\textwidth}
\vspace{10pt}
{\footnotesize
\begin{tabular}{llr}
Method & CDE Loss $\pm$ SE & Storage (single CDE)\\
\hline
Marginal & -3.192 $\pm$ 0.007 & 200 floats (1.6 KB)\\
RFCDE-Limited & -11.038 $\pm$ 0.047 & 200 floats (1.6 KB)\\
NNKCDE & -12.139 $\pm$ 0.043 & 200 floats (1.6 KB)\\
FlexZBoost & -12.272 $\pm$ 0.044 & 30 coeffs (0.24 KB)\\
DeepCDE & -12.821 $\pm$ 0.04 & 31 coeffs (0.25 KB)\\
RFCDE & -13.108 $\pm$ 0.052 & 200 floats (1.6 KB)
\end{tabular}
}
\end{minipage}
\caption{Method comparison via the CDE loss of Equation~\ref{eq:estimated-cde-loss} with estimated standard error (SE) and the storage space for each galaxy's \pz\ CDE.  
As \flexzboost\ and \deepcde\ are basis expansion methods, we need only to store the estimated coefficients for a lossless compression of the CDEs; 
the other CDEs are discretized to 200 bins.}\label{tab:teddy-cde-loss}
\end{figure}

Figure~\ref{tab:teddy-cde-loss} presents the estimated CDE loss of Equation~\ref{eq:estimated-cde-loss} with estimated standard error (SE), as well as the storage space for each galaxy's \pz\ CDE.
Figure \ref{fig:teddy-pvals} shows that the different models, including the clearly misspecified ``Marginal'' model, achieve comparable performance on goodness-of-fit diagnostics.
However, Figure~\ref{tab:teddy-cde-loss} shows the discriminatory power of the CDE loss function, which distinguishes the methods from one another.
This emphasizes the need for method comparison through loss functions in addition to goodness-of-fit diagnostics.
We note that the ranking in CDE loss also correlates roughly with the quality of the point estimates, as shown in Appendix \ref{sec:app_pointestimates}. 
The \lsst -DESC PZ DC1 paper \citep{desc_photoz} draws similar conclusions from a comprehensive photo-$z$ code comparison where the ``Marginal'' model (there referred to as \texttt{trainZ} photo-$z$ PDF estimator, the experimental control) outperformed all codes when using traditional metrics for assessing photo-$z$ PDF acccuracy. 
Indeed, of the metrics considered in DC1, the CDE loss was the only metric that could appropriately penalize the pathological \texttt{trainZ}.

\subsection{Likelihood-free cosmological inference: Multivariate response}
\label{sec:mult_response_mult_data}

To showcase the ability to target joint distributions, we apply \nnkcde\ and \rfcde\ to the problem of estimating multivariate posteriors $p(\boldsymbol{\theta} | \mathbf{x}_{\rm obs})$ of the cosmological parameters in a likelihood-free setting via ABC-CDE~\citep{izbicki2019abc}. 

ABC is an approach to parameter inference in settings where the likelihood is not tractable but we have a forward model that can simulate data $\mathbf{x}$ under fixed parameter settings $\boldsymbol{\theta}$. 
The simplest form of ABC is the ABC rejection sampling algorithm, where a set of parameters is first drawn from a prior distribution. 
The simulated data $\mathbf{x}$ is accepted with tolerance $\varepsilon \geq 0$ if $d(\mathbf{x},\mathbf{x}_{\rm obs}) \leq \varepsilon$ for some distance metric $d$ (e.g., the Euclidean distance) that measures the discrepancy between the simulated data $\mathbf{x}$ and observed data $\mathbf{x}_{\rm obs}$. 
The outcome of the ABC rejection algorithm for small enough $\epsilon$ is a sample of parameter values approximately distributed according to the desired posterior distribution $p(\boldsymbol{\theta} | \mathbf{x}_{\rm obs})$. 

The basic idea of ABC-CDE is to improve the ABC estimate --- and hence reduce the number of required simulations --- by using the ABC sample/output as input to a CDE method tuned with the CDE loss restricted to a neighborhood of $\xobs$ defined by the tolerance $\epsilon$ (\cite{izbicki2019abc}, Equation 3). 
Hence, in ABC-CDE, our CDE method can be seen as a post-adjustment method: 
it returns an estimate $\widehat{p}(\boldsymbol{\theta} | \mathbf{x})$ which we evaluate at the point $\mathbf{x}=\mathbf{x}_{\rm obs}$ to obtain a more accurate approximation of $p(\boldsymbol{\theta} | \mathbf{x}_{\rm obs})$. 
This could also be beneficial in an active learning setting where the posterior distribution is used to identify relevant regions of the parameter space (e.g., \citep{lueckmann2019likelihood, papamakarios2019sequential, alsing2019fast}).

In this example, we consider the problem of cosmological parameter inference via {\em cosmic shear}, caused by weak gravitational lensing inducing a distortion in images of distant galaxies. 
The size and direction of the distortion is directly related to the size and shape of the matter distribution along the line of sight, which varies across the universe.
We use shear correlation functions to constrain the dark matter density \(\Omega_{M}\) and matter power spectrum normalization \(\sigma_{8}\) parameters of the \(\Lambda\)CDM cosmological model, which predicts the properties and evolution of the large scale structure of the matter distribution.
For further background see \cite{hoekstra2008weak}, \cite{munshi2008cosmology} and \citet{mandelbaum2017weak}.
Here we use the \code{GalSim}\footnote{\url{https://github.com/GalSim-developers/GalSim}} toolkit \citep{rowe2015galsim} to generate simplified galaxy shears distributed according to a Gaussian random field determined by $(\Omega_{M}, \sigma_{8})$.
The binned shear correlation functions serve as our input data or summary statistics $x$.
For the inference, we assume uniform priors \(\Omega_{M} \sim U(0.1, 0.8)\) and \(\sigma_{8} \sim U(0.5, 1.0)\) and fix \(h = 0.7\), \(\Omega_{b} = 0.045\), \(z = 0.7\). 

\begin{figure}[htbp]
\centering
\includegraphics[width=0.8\textwidth]{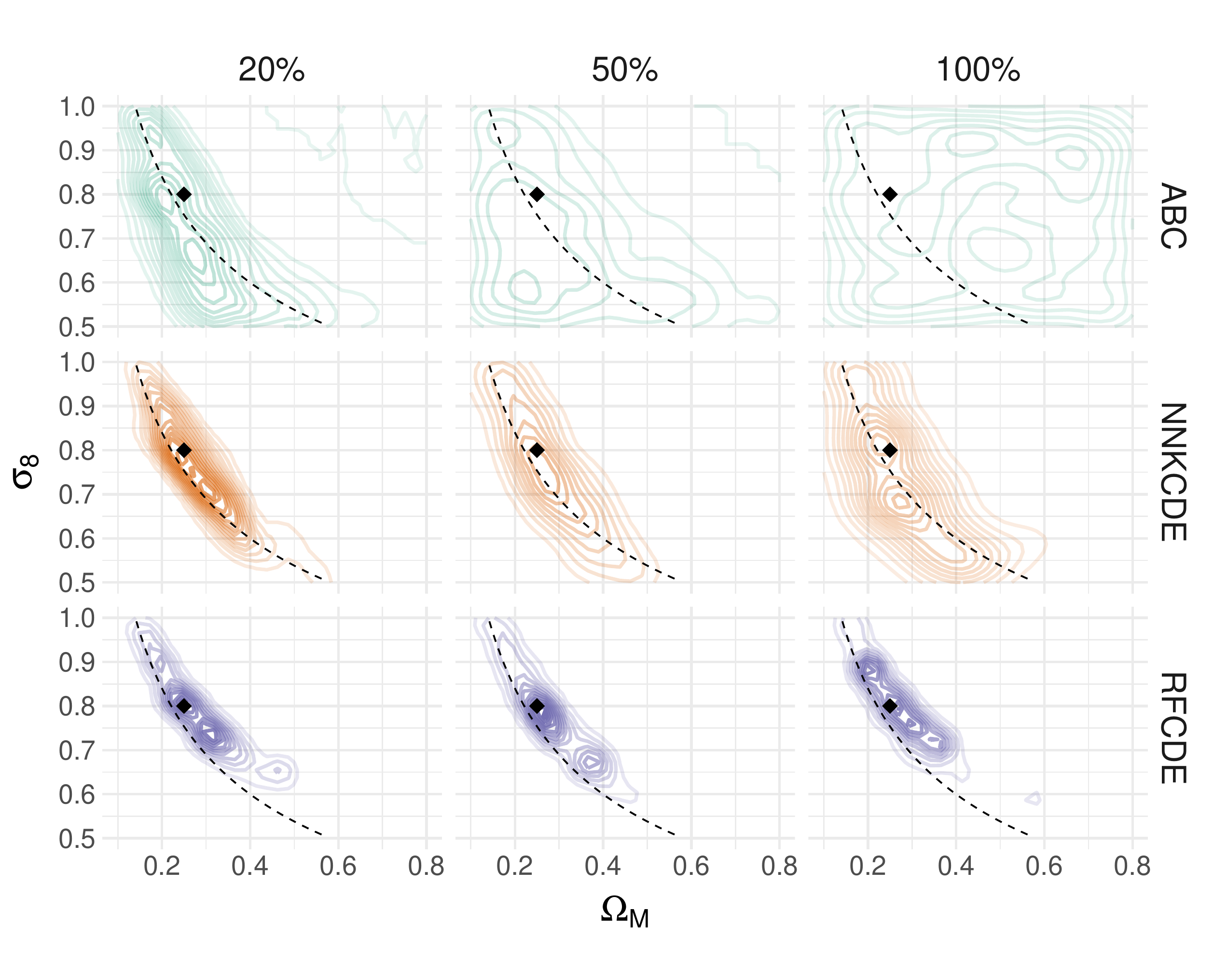}
\caption{\label{fig:orge983804} 
 Cosmological parameter inference in an LFI setting with simulated weak lensing data.  
 The top row shows the estimated bivariate distribution of $\Omega_M$ and $\sigma_8$ for ABC rejection sampling at different acceptance rates (20\%, 50\%, and 100\%). 
 The middle row shows the estimated posterior densities after applying \nnkcde\ to the ABC sample, and the bottom row when applying \rfcde.\ 
 Both \nnkcde\ and \rfcde\ tuned with a CDE loss improve on ABC across acceptance rate levels, i.e., they provide approximate posteriors that are more concentrated around the true observed parameter values and the degeneracy curve on which the data are indistinguishable (shown here as a black diamond and dashed line, respectively). 
 We even see some structure at an ABC acceptance rate of 1 (right column); 
 that is, at an ABC threshold of $\epsilon \rightarrow \infty$ for which the entire sample is accepted by ABC and passed to our CDE code.
 HPD values were not included as they fail to distinguish between conditional and marginal distribution, as mentioned above.}
\end{figure}

The top row of Figure~\ref{fig:orge983804} shows the estimated bivariate posterior distribution of $\boldsymbol{\theta}=(\Omega_{M},\sigma_{8})$ from ABC rejection sampling only, at varying acceptance rates (20\%, 50\%, and 100\%). 
An acceptance rate of 100\% just returns the (uniform) ABC prior distribution, whereas the ABC posteriors for smaller acceptance rates (that is, smaller values of $\varepsilon$) concentrate around the parameter degeneracy curve (shown as a dashed line) on which the data are indistinguishable. 
The second and third rows show the estimated posteriors when we, respectively, improve the initial ABC estimate by applying \nnkcde\ and \rfcde\ tuned with our CDE surrogate loss. 
A 
result that is apparent from the figure
is that \nnkcde\ and \rfcde\ fitted with a CDE loss are able to capture the degeneracy curve at a larger acceptance rate, that is, for a smaller number of simulations, than when using ABC only. 

\begin{table}[htbp]
  \label{tab:lfi-experiment-cdeloss}
  \caption{\small Performance of ABC, \nnkcde, \rfcde\ in LFI settings with simulated weak lensing data in terms of the surrogate CDE loss.}
  \centering
  \begin{tabular}{lrrr}
    \hline
    & \multicolumn{3}{c}{CDE Loss $\pm$ SE ($\times 10^{-5}$)} \\ \hline
    Method \textbackslash Acceptance Rate & 20\% & 50\% & 100\% \\ \hline
     ABC & -0.686 $\pm$ 0.009 & -0.392 $\pm$ 0.004 & -0.227 $\pm$ 0.001 \\
     \nnkcde\ & -1.652 $\pm$ 0.022 & -1.199 $\pm$ 0.016 & -0.844 $\pm$ 0.010\\
     \rfcde\ & -4.129 $\pm$ 0.063 & -3.698 $\pm$ 0.064 & -2.817 $\pm$ 0.055 \\
  \end{tabular}
\end{table}

We also have a direct measure of performance via the CDE loss and can adapt to different types of data by leveraging different CDE codes.
Table \ref{tab:lfi-experiment-cdeloss} shows how the methods compare in terms of the surrogate loss.
While decreasing the acceptance rate benefits all methods, it is clear that CDE-based approaches have better performance in all cases.

\subsection{Spec-$z$ estimation: Functional input}
\label{sec:functional}
In this example, we compare CDE methods in the context of spectroscopic redshift prediction using 2812 high-resolution spectra from the Sloan Digital Sky Survey (\textsc{SDSS}) Data Release 6 (preprocessed with the cuts of \citealt{Richards:EtAl:2009}), corresponding to features $\x$ of flux measurements at $d=3501$ wavelengths.  
The high-resolution spectra $\x$ can be seen as functional inputs on a continuum of wavelength values. 
Spectroscopic redshifts (or redshifts $z$ predicted from spectra $\x$) tend to be both accurate and precise, so the density $p(z \gvn \x)$ is well-approximated by a delta function at the true redshift. 
For the purposes of illustrating the use of the CDE codes, we define a noisified redshift $z_i= z_i^{\rm SDSS}+\epsilon_i$, where $\epsilon_i$ are independent and identically distributed variables drawn from a normal distribution $N(0,0 .02)$ and $z_i^{\rm SDSS}$ is the true redshift of galaxy $i$ provided by \textsc{SDSS}. 
Thus the conditional density $p(z_i \gvn \x_i)$ of this example is a Gaussian distribution with mean $z_i^{\rm SDSS}$ and variance 0.02.

We compare \frfcde, a ``Functional'' adaptation of \rfcde, with a standard ``Vector'' implementation of \rfcde, which treats functional data as a vector.
For completeness, we also compare against standard regression random forest combined with KDE, as well as \flexcode\code{-Spec}, an extension of \flexcode\ with a Spectral Series regression~\citep{Richards:EtAl:2009,Freeman, lee2016spectral} for determining the expansion coefficients. 
We train on 2000 galaxies and test on the remaining galaxies.
The \rfcde\ and the \frfcde\ trees are both trained with $n_{trees} = 1000$, $n_{basis}  = 31$, and bandwidths chosen by plug-in estimators.
We use \(\lambda=50\) for the \frfcde\ rate parameter of the Poisson process that defines the variable groupings. 

\begin{table}[htbp]
  \label{tab:frfcde-experiment}
  \caption{\small Performance of \rfcde\ for functional data. 
  We achieve both a lower CDE loss and computational time for Functional \rfcde (as compared to Vector \rfcde) by leveraging the functional nature of the data.}
  \centering
  \begin{tabular}{lrr}
    Method & Train Time (in sec) & CDE Loss $\pm$ SE\\
    \hline
     Functional \rfcde\ & 24.89 & -3.38 $\pm$ 0.155\\
     Vector \rfcde\ & 41.60 & -2.52 $\pm$ 0.104\\
     Regression RF + KDE & 50.63 & -3.42 $\pm$ 0.120 \\
     \flexcode\code{-Spec} & 4017.93 & -3.53 $\pm$ 0.136 \\
  \end{tabular}
\end{table}

Table \ref{tab:frfcde-experiment} contains the CDE loss and train time for both the vector-based and functional-based \rfcde\ models on the \textsc{SDSS} data, as well as for \flexcode\code{-Spec}.
We obtain substantial gains when incorporating functional features both in terms of CDE loss and computational time.
The computational gains are attributed to requiring fewer searches for each split point as the default value of $m_{try} = \sqrt{d}$ is reduced.
As anticipated, vector-based \rfcde\ underperforms with functional data due to the tree splits on CDE loss struggling to pick up signals in the data and returning almost the same conditional density regardless of the input. 
In contrast, splitting on mean squared error is an easier task and we include the results of regular regression random forest and KDE, which are comparable to the results of \frfcde.
Whereas random forests rely on variable selection (either individual variables as in vectorial \rfcde\ or grouped together as in \frfcde) for dimensional reduction, \flexcode\code{-Spec} is based on the Spectral Series mechanism for dimension reduction that finds (potentially) nonlinear, sparse structure in the data distribution~\citep{izbicki2016nonparametric}.  
Both types of dimension reduction are reasonable for spec-$z$ estimation, as particular wavelength locations or regions in the galaxy SED could carry information of the galaxy's true redshift; 
\rfcde\ and \frfcde\ are effective in finding such locations by variable selection. 
Spectral Series on the other hand are able to recover low-dimensional manifold structure in the entire data ensemble; 
as \citealt{Richards:EtAl:2009} show, the main direction of variation in SDSS spectra is directly related to the spectroscopic redshift.

\section{Conclusions} 
This paper presents statistical tools and software for uncertainty quantification in complex regression and parameter inference tasks.
Given a set of features $\x$ with associated response $\y$, our methods extend the usual point prediction task of classification and regression to nonparametric estimation of the entire (conditional) probability density $p(\y \gvn \x)$. 
The described CDE methods are meant to handle a range of different data settings as outlined in Table~\ref{tab:method-strenghts}. This paper includes examples 
of code usage in the contexts of \pz\ estimation, likelihood-free inference for cosmology analysis, and spec-z estimation. 
In addition, it provides tools for CDE method assessment and for choosing tuning parameters of the CDE methods in a principled way.

Our software includes four packages for CDE --- \nnkcde, \rfcde, \flexcode\ and \deepcde --- each using a different machine learning algorithm, as well as a package for  model assessment, \cdetools.  All packages are implemented in both \python\ and \rlang, with no dependence on proprietary software, making them compatible with any operating system supporting either language (e.g., Windows, MacOS, Ubuntu, and any other Unix-based OS).
The code is provided in publicly available Github repositories, documented using \python\ and \rlang\ function documentation standards, and equipped with \code{Travis CI}\footnote{\url{https://docs.travis-ci.com/}} automatic bug-tracking.
Finally, our software makes uncertainty quantification straightforward for those used to standard open-source machine learning \python\ packages: 
\nnkcde, \rfcde, {\flexcode} share the \code{sklearn} API (with \texttt{fit} and \texttt{predict} methods), which makes our methods compatible with \code{sklearn} wrapper functions (for e.g., cross-validation and model ensembling). 
\deepcde\ has implementations for both \code{Tensorflow} and \code{Pytorch}, two of the most widely used deep learning frameworks, and can easily be combined with most existing network architectures. 


\section*{Acknowledgements}

We would like to thank the two anonymous reviewers for their insightful comments that helped improve the manuscript.
ND, TP and ABL were partially supported by the National Science Foundation under Grant Nos.\ DMS1521786.
RI was supported by Conselho Nacional de Desenvolvimento Cient\'ifico e Tecnol\'ogico (grant number 306943/2017-4) and Funda\c{c}\~ao de Amparo \`a Pesquisa do Estado de S\~ao Paulo (grants number 2017/03363-8
and 2019/11321-9).
AIM acknowledges support from the Max Planck Society and the Alexander von Humboldt Foundation in the framework of the Max Planck-Humboldt Research Award endowed by the Federal Ministry of Education and Research.
During the completion of this work, AIM was advised by David W. Hogg and supported in part by National Science Foundation grant AST-1517237.

\appendix

\section{Examples of code usage}\label{sec:code_examples}

In this section we provide examples of code usage in \python\ and \rlang. 
To save space, we have included the \deepcde\ examples directly in the code release\footnote{\url{https://github.com/tpospisi/DeepCDE/tree/master/model_examples}} \citep{DeepCDE_code} rather than in this section. 

\subsection{\nnkcde}
\label{app:nnkcde}

\noindent\rule{\textwidth}{0.5pt}
\begin{Verbatim}[commandchars=\\\{\},codes={\catcode`\$=3\catcode`\^=7\catcode`\_=8}]
\PYG{k+kn}{import} \PYG{n+nn}{numpy} \PYG{k+kn}{as} \PYG{n+nn}{np}
\PYG{k+kn}{import} \PYG{n+nn}{nnkcde}

\PYG{c+c1}{\PYGZsh{} Fit the model}
\PYG{n}{model} \PYG{o}{=} \PYG{n}{nnkcde}\PYG{o}{.}\PYG{n}{NNKCDE}\PYG{p}{()}
\PYG{n}{model}\PYG{o}{.}\PYG{n}{fit}\PYG{p}{(}\PYG{n}{x\PYGZus{}train}\PYG{p}{,} \PYG{n}{y\PYGZus{}train}\PYG{p}{)}

\PYG{c+c1}{\PYGZsh{} Tune parameters: bandwidth, number of neighbors k}
\PYG{c+c1}{\PYGZsh{} Pick best loss on validation data}
\PYG{n}{k\PYGZus{}choices} \PYG{o}{=} \PYG{p}{[}\PYG{l+m+mi}{5}\PYG{p}{,} \PYG{l+m+mi}{10}\PYG{p}{,} \PYG{l+m+mi}{100}\PYG{p}{,} \PYG{l+m+mi}{1000}\PYG{p}{]}
\PYG{n}{bandwith\PYGZus{}choices} \PYG{o}{=} \PYG{p}{[}\PYG{l+m+mf}{0.01}\PYG{p}{,} \PYG{l+m+mf}{0.05}\PYG{p}{,} \PYG{l+m+mf}{0.1}\PYG{p}{]}
\PYG{n}{model}\PYG{o}{.}\PYG{n}{tune}\PYG{p}{(}\PYG{n}{x\PYGZus{}validation}\PYG{p}{,} \PYG{n}{y\PYGZus{}validation}\PYG{p}{,} \PYG{n}{k\PYGZus{}choices}\PYG{p}{,} \PYG{n}{bandwidth\PYGZus{}choices}\PYG{p}{)}

\PYG{c+c1}{\PYGZsh{} Predict new densities on grid}
\PYG{n}{y\PYGZus{}grid} \PYG{o}{=} \PYG{n}{np}\PYG{o}{.}\PYG{n}{linspace}\PYG{p}{(}\PYG{n}{y\PYGZus{}min}\PYG{p}{,} \PYG{n}{y\PYGZus{}max}\PYG{p}{,} \PYG{n}{n\PYGZus{}grid}\PYG{p}{)}
\PYG{n}{cde\PYGZus{}test} \PYG{o}{=} \PYG{n}{model}\PYG{o}{.}\PYG{n}{predict}\PYG{p}{(}\PYG{n}{x\PYGZus{}test}\PYG{p}{,} \PYG{n}{y\PYGZus{}grid}\PYG{p}{)}
\end{Verbatim}
\noindent\rule{\textwidth}{0.5pt}

\vspace{0.1cm}

\noindent\rule{\textwidth}{0.5pt}
\begin{Verbatim}[commandchars=\\\{\},codes={\catcode`\$=3\catcode`\^=7\catcode`\_=8}]
\PYG{n+nf}{library}\PYG{p}{(}\PYG{n}{NNKCDE}\PYG{p}{)}

\PYG{c+c1}{\PYGZsh{} Fit the model}
\PYG{n}{model} \PYG{o}{\PYGZlt{}\PYGZhy{}} \PYG{n}{NNKCDE}\PYG{o}{::}\PYG{n}{NNKCDE}\PYG{o}{\PYGZdl{}}\PYG{n+nf}{new}\PYG{p}{(}\PYG{n}{x\PYGZus{}train}\PYG{p}{,} \PYG{n}{y\PYGZus{}train}\PYG{p}{)}

\PYG{c+c1}{\PYGZsh{} Tune parameters: bandwidth, number of neighbors k}
\PYG{c+c1}{\PYGZsh{} Pick best loss on validation data}
\PYG{n}{k\PYGZus{}choices} \PYG{o}{\PYGZlt{}\PYGZhy{}} \PYG{n+nf}{c}\PYG{p}{(}\PYG{l+m}{5}\PYG{p}{,} \PYG{l+m}{10}\PYG{p}{,} \PYG{l+m}{100}\PYG{p}{,} \PYG{l+m}{1000}\PYG{p}{)}
\PYG{n}{bandwith\PYGZus{}choices} \PYG{o}{\PYGZlt{}\PYGZhy{}} \PYG{n+nf}{c}\PYG{p}{(}\PYG{l+m}{0.01}\PYG{p}{,} \PYG{l+m}{0.05}\PYG{p}{,} \PYG{l+m}{0.1}\PYG{p}{)}
\PYG{n}{model}\PYG{o}{\PYGZdl{}}\PYG{n+nf}{tune}\PYG{p}{(}\PYG{n}{x\PYGZus{}validation}\PYG{p}{,} \PYG{n}{y\PYGZus{}validation}\PYG{p}{,} \PYG{n}{k\PYGZus{}grid} \PYG{o}{=} \PYG{n}{k\PYGZus{}choices}\PYG{p}{,} \PYG{n}{h\PYGZus{}grid} \PYG{o}{=} \PYG{n}{bandwidth\PYGZus{}choices}\PYG{p}{)}

\PYG{c+c1}{\PYGZsh{} Predict new densities on grid}
\PYG{n}{y\PYGZus{}grid} \PYG{o}{\PYGZlt{}\PYGZhy{}} \PYG{n+nf}{seq}\PYG{p}{(}\PYG{n}{y\PYGZus{}min}\PYG{p}{,} \PYG{n}{y\PYGZus{}max}\PYG{p}{,} \PYG{n}{length.out} \PYG{o}{=} \PYG{n}{n\PYGZus{}grid}\PYG{p}{)}
\PYG{n}{cde\PYGZus{}test} \PYG{o}{\PYGZlt{}\PYGZhy{}} \PYG{n}{model}\PYG{o}{\PYGZdl{}}\PYG{n+nf}{predict}\PYG{p}{(}\PYG{n}{x\PYGZus{}test}\PYG{p}{,} \PYG{n}{y\PYGZus{}grid}\PYG{p}{)}
\end{Verbatim}
\noindent\rule{\textwidth}{0.5pt}

\bigskip

\subsection{\rfcde}
\label{app:rfcde}

\noindent\rule{\textwidth}{0.5pt}
\begin{Verbatim}[commandchars=\\\{\},codes={\catcode`\$=3\catcode`\^=7\catcode`\_=8}]
\PYG{k+kn}{import} \PYG{n+nn}{numpy} \PYG{k+kn}{as} \PYG{n+nn}{np}
\PYG{k+kn}{import} \PYG{n+nn}{rfcde}

\PYG{c+c1}{\PYGZsh{} Parameters}
\PYG{n}{n\PYGZus{}trees} \PYG{o}{=} \PYG{l+m+mi}{1000}   \PYG{c+c1}{\PYGZsh{} Number of trees in the forest}
\PYG{n}{mtry} \PYG{o}{=} \PYG{l+m+mi}{4}         \PYG{c+c1}{\PYGZsh{} Number of variables to potentially split at in each node}
\PYG{n}{node\PYGZus{}size} \PYG{o}{=} \PYG{l+m+mi}{20}   \PYG{c+c1}{\PYGZsh{} Smallest node size}
\PYG{n}{n\PYGZus{}basis} \PYG{o}{=} \PYG{l+m+mi}{15}     \PYG{c+c1}{\PYGZsh{} Number of basis functions}
\PYG{n}{bandwidth} \PYG{o}{=} \PYG{l+m+mf}{0.2}  \PYG{c+c1}{\PYGZsh{} Kernel bandwith \PYGZhy{} used for prediction only}

\PYG{c+c1}{\PYGZsh{} Fit the model}
\PYG{n}{forest} \PYG{o}{=} \PYG{n}{rfcde}\PYG{o}{.}\PYG{n}{RFCDE}\PYG{p}{(}\PYG{n}{n\PYGZus{}trees}\PYG{o}{=}\PYG{n}{n\PYGZus{}trees}\PYG{p}{,} \PYG{n}{mtry}\PYG{o}{=}\PYG{n}{mtry}\PYG{p}{,} \PYG{n}{node\PYGZus{}size}\PYG{o}{=}\PYG{n}{node\PYGZus{}size}\PYG{p}{,} \PYG{n}{n\PYGZus{}basis}\PYG{o}{=}\PYG{n}{n\PYGZus{}basis}\PYG{p}{)}
\PYG{n}{forest}\PYG{o}{.}\PYG{n}{train}\PYG{p}{(}\PYG{n}{x\PYGZus{}train}\PYG{p}{,} \PYG{n}{y\PYGZus{}train}\PYG{p}{)}

\PYG{c+c1}{\PYGZsh{} Predict new densities on grid}
\PYG{n}{y\PYGZus{}grid} \PYG{o}{=} \PYG{n}{np}\PYG{o}{.}\PYG{n}{linspace}\PYG{p}{(}\PYG{l+m+mi}{0}\PYG{p}{,} \PYG{l+m+mi}{1}\PYG{p}{,} \PYG{n}{n\PYGZus{}grid}\PYG{p}{)}
\PYG{n}{cde\PYGZus{}test} \PYG{o}{=} \PYG{n}{forest}\PYG{o}{.}\PYG{n}{predict}\PYG{p}{(}\PYG{n}{x\PYGZus{}test}\PYG{p}{,} \PYG{n}{y\PYGZus{}grid}\PYG{p}{,} \PYG{n}{bandwidth}\PYG{p}{)}

\PYG{c+c1}{\PYGZsh{} Predict conditional means for CDE\PYGZhy{}optimized forest}
\PYG{n}{cond\PYGZus{}mean\PYGZus{}test} \PYG{o}{=} \PYG{n}{forest}\PYG{o}{.}\PYG{n}{predict\PYGZus{}mean}\PYG{p}{(}\PYG{n}{x\PYGZus{}test}\PYG{p}{)}

\PYG{c+c1}{\PYGZsh{} Predict quantiles (i.e., quantile regression) for CDE\PYGZhy{}optimized forest}
\PYG{n}{alpha\PYGZus{}quantile} \PYG{o}{=} \PYG{o}{.}\PYG{l+m+mi}{90}
\PYG{n}{quant\PYGZus{}test} \PYG{o}{=} \PYG{n}{forest}\PYG{o}{.}\PYG{n}{predict\PYGZus{}quantile}\PYG{p}{(}\PYG{n}{x\PYGZus{}test}\PYG{p}{,} \PYG{n}{alpha\PYGZus{}quantile}\PYG{p}{)}
\end{Verbatim}
\noindent\rule{\textwidth}{0.5pt}

\vspace{0.1cm}

\noindent\rule{\textwidth}{0.5pt}
\begin{Verbatim}[commandchars=\\\{\},codes={\catcode`\$=3\catcode`\^=7\catcode`\_=8}]
\PYG{n+nf}{library}\PYG{p}{(}\PYG{n}{RFCDE}\PYG{p}{)}

\PYG{c+c1}{\PYGZsh{} Parameters}
\PYG{n}{n\PYGZus{}trees} \PYG{o}{\PYGZlt{}\PYGZhy{}} \PYG{l+m}{1000}   \PYG{c+c1}{\PYGZsh{} Number of trees in the forest}
\PYG{n}{mtry} \PYG{o}{\PYGZlt{}\PYGZhy{}} \PYG{l+m}{4}         \PYG{c+c1}{\PYGZsh{} Number of variables to potentially split at in each node}
\PYG{n}{node\PYGZus{}size} \PYG{o}{\PYGZlt{}\PYGZhy{}} \PYG{l+m}{20}   \PYG{c+c1}{\PYGZsh{} Smallest node size}
\PYG{n}{n\PYGZus{}basis} \PYG{o}{\PYGZlt{}\PYGZhy{}} \PYG{l+m}{15}     \PYG{c+c1}{\PYGZsh{} Number of basis functions}
\PYG{n}{bandwidth} \PYG{o}{\PYGZlt{}\PYGZhy{}} \PYG{l+m}{0.2}  \PYG{c+c1}{\PYGZsh{} Kernel bandwith \PYGZhy{} used for prediction only}

\PYG{c+c1}{\PYGZsh{} Fit the model}
\PYG{n}{forest} \PYG{o}{\PYGZlt{}\PYGZhy{}} \PYG{n}{RFCDE}\PYG{o}{::}\PYG{n+nf}{RFCDE}\PYG{p}{(}\PYG{n}{x\PYGZus{}train}\PYG{p}{,} \PYG{n}{y\PYGZus{}train}\PYG{p}{,} \PYG{n}{n\PYGZus{}trees} \PYG{o}{=} \PYG{n}{n\PYGZus{}trees}\PYG{p}{,} \PYG{n}{mtry} \PYG{o}{=} \PYG{n}{mtry}\PYG{p}{,}
                      \PYG{n}{node\PYGZus{}size} \PYG{o}{=} \PYG{n}{node\PYGZus{}size}\PYG{p}{,} \PYG{n}{n\PYGZus{}basis} \PYG{o}{=} \PYG{n}{n\PYGZus{}basis}\PYG{p}{)}

\PYG{c+c1}{\PYGZsh{} Predict new densities on grid}
\PYG{n}{y\PYGZus{}grid} \PYG{o}{\PYGZlt{}\PYGZhy{}} \PYG{n+nf}{seq}\PYG{p}{(}\PYG{l+m}{0}\PYG{p}{,} \PYG{l+m}{1}\PYG{p}{,} \PYG{n}{length.out} \PYG{o}{=} \PYG{n}{n\PYGZus{}grid}\PYG{p}{)}
\PYG{n}{cde\PYGZus{}test} \PYG{o}{\PYGZlt{}\PYGZhy{}} \PYG{n+nf}{predict}\PYG{p}{(}\PYG{n}{forest}\PYG{p}{,} \PYG{n}{x\PYGZus{}test}\PYG{p}{,} \PYG{n}{y\PYGZus{}grid}\PYG{p}{,} \PYG{n}{response} \PYG{o}{=} \PYG{l+s}{\PYGZsq{}CDE\PYGZsq{}}\PYG{p}{,} \PYG{n}{bandwidth} \PYG{o}{=} \PYG{n}{bandwidth}\PYG{p}{)}

\PYG{c+c1}{\PYGZsh{} Predict conditional means for CDE\PYGZhy{}optimized forest}
\PYG{n}{cond\PYGZus{}mean\PYGZus{}test} \PYG{o}{\PYGZlt{}\PYGZhy{}} \PYG{n+nf}{predict}\PYG{p}{(}\PYG{n}{forest}\PYG{p}{,} \PYG{n}{x\PYGZus{}test}\PYG{p}{,} \PYG{n}{y\PYGZus{}grid}\PYG{p}{,} \PYG{n}{response} \PYG{o}{=} \PYG{l+s}{\PYGZsq{}mean\PYGZsq{}}\PYG{p}{,} \PYG{n}{bandwidth} \PYG{o}{=} \PYG{n}{bandwidth}\PYG{p}{)}

\PYG{c+c1}{\PYGZsh{} Predict quantiles (i.e., quantile regression) for CDE\PYGZhy{}optimized forest}
\PYG{n}{alpha\PYGZus{}quantile} \PYG{o}{\PYGZlt{}\PYGZhy{}} \PYG{l+m}{.90}
\PYG{n}{quant\PYGZus{}test} \PYG{o}{\PYGZlt{}\PYGZhy{}} \PYG{n+nf}{predict}\PYG{p}{(}\PYG{n}{forest}\PYG{p}{,} \PYG{n}{x\PYGZus{}test}\PYG{p}{,} \PYG{n}{y\PYGZus{}grid}\PYG{p}{,} \PYG{n}{response} \PYG{o}{=} \PYG{l+s}{\PYGZsq{}quantile\PYGZsq{}}\PYG{p}{,}
                     \PYG{n}{quantile} \PYG{o}{=} \PYG{n}{alpha\PYGZus{}quantile}\PYG{p}{,} \PYG{n}{bandwidth} \PYG{o}{=} \PYG{n}{bandwidth}\PYG{p}{)}
\end{Verbatim}
\noindent\rule{\textwidth}{0.5pt}

\bigskip

\subsubsection{\frfcde}
\label{app:frfcde}

\noindent\rule{\textwidth}{0.5pt}
\begin{Verbatim}[commandchars=\\\{\},codes={\catcode`\$=3\catcode`\^=7\catcode`\_=8}]
\PYG{k+kn}{import} \PYG{n+nn}{numpy} \PYG{k+kn}{as} \PYG{n+nn}{np}
\PYG{k+kn}{import} \PYG{n+nn}{rfcde}

\PYG{c+c1}{\PYGZsh{} Parameters}
\PYG{n}{n\PYGZus{}trees} \PYG{o}{=} \PYG{l+m+mi}{1000}     \PYG{c+c1}{\PYGZsh{} Number of trees in the forest}
\PYG{n}{mtry} \PYG{o}{=} \PYG{l+m+mi}{4}           \PYG{c+c1}{\PYGZsh{} Number of variables to potentially split at in each node}
\PYG{n}{node\PYGZus{}size} \PYG{o}{=} \PYG{l+m+mi}{20}     \PYG{c+c1}{\PYGZsh{} Smallest node size}
\PYG{n}{n\PYGZus{}basis} \PYG{o}{=} \PYG{l+m+mi}{15}       \PYG{c+c1}{\PYGZsh{} Number of basis functions}
\PYG{n}{bandwidth} \PYG{o}{=} \PYG{l+m+mf}{0.2}    \PYG{c+c1}{\PYGZsh{} Kernel bandwith \PYGZhy{} used for prediction only}
\PYG{n}{lambda\PYGZus{}param} \PYG{o}{=} \PYG{l+m+mi}{10}  \PYG{c+c1}{\PYGZsh{} Poisson Process parameter}

\PYG{c+c1}{\PYGZsh{} Fit the model}
\PYG{n}{functional\PYGZus{}forest} \PYG{o}{=} \PYG{n}{rfcde}\PYG{o}{.}\PYG{n}{RFCDE}\PYG{p}{(}\PYG{n}{n\PYGZus{}trees}\PYG{o}{=}\PYG{n}{n\PYGZus{}trees}\PYG{p}{,} \PYG{n}{mtry}\PYG{o}{=}\PYG{n}{mtry}\PYG{p}{,} \PYG{n}{node\PYGZus{}size}\PYG{o}{=}\PYG{n}{node\PYGZus{}size}\PYG{p}{,}
                                \PYG{n}{n\PYGZus{}basis}\PYG{o}{=}\PYG{n}{n\PYGZus{}basis}\PYG{p}{)}
\PYG{n}{functional\PYGZus{}forest}\PYG{o}{.}\PYG{n}{train}\PYG{p}{(}\PYG{n}{x\PYGZus{}train}\PYG{p}{,} \PYG{n}{y\PYGZus{}train}\PYG{p}{,} \PYG{n}{flamba}\PYG{o}{=}\PYG{n}{lambda\PYGZus{}param}\PYG{p}{)}

\PYG{c+c1}{\PYGZsh{} ... Same as RFCDE for prediction ...}
\end{Verbatim}
\noindent\rule{\textwidth}{0.5pt}

\vspace{0.1cm}

\noindent\rule{\textwidth}{0.5pt}
\begin{Verbatim}[commandchars=\\\{\},codes={\catcode`\$=3\catcode`\^=7\catcode`\_=8}]
\PYG{n+nf}{library}\PYG{p}{(}\PYG{n}{RFCDE}\PYG{p}{)}

\PYG{c+c1}{\PYGZsh{} Parameters}
\PYG{n}{n\PYGZus{}trees} \PYG{o}{\PYGZlt{}\PYGZhy{}} \PYG{l+m}{1000}     \PYG{c+c1}{\PYGZsh{} Number of trees in the forest}
\PYG{n}{mtry} \PYG{o}{\PYGZlt{}\PYGZhy{}} \PYG{l+m}{4}           \PYG{c+c1}{\PYGZsh{} Number of variables to potentially split at in each node}
\PYG{n}{node\PYGZus{}size} \PYG{o}{\PYGZlt{}\PYGZhy{}} \PYG{l+m}{20}     \PYG{c+c1}{\PYGZsh{} Smallest node size}
\PYG{n}{n\PYGZus{}basis} \PYG{o}{\PYGZlt{}\PYGZhy{}} \PYG{l+m}{15}       \PYG{c+c1}{\PYGZsh{} Number of basis functions}
\PYG{n}{bandwidth} \PYG{o}{\PYGZlt{}\PYGZhy{}} \PYG{l+m}{0.2}    \PYG{c+c1}{\PYGZsh{} Kernel bandwith \PYGZhy{} used for prediction only}
\PYG{n}{lambda\PYGZus{}param} \PYG{o}{\PYGZlt{}\PYGZhy{}} \PYG{l+m}{10}  \PYG{c+c1}{\PYGZsh{} Poisson Process parameter}

\PYG{c+c1}{\PYGZsh{} Fit the model}
\PYG{n}{functional\PYGZus{}forest} \PYG{o}{\PYGZlt{}\PYGZhy{}} \PYG{n}{RFCDE}\PYG{o}{::}\PYG{n+nf}{RFCDE}\PYG{p}{(}\PYG{n}{x\PYGZus{}train}\PYG{p}{,} \PYG{n}{y\PYGZus{}train}\PYG{p}{,} \PYG{n}{n\PYGZus{}trees} \PYG{o}{=} \PYG{n}{n\PYGZus{}trees}\PYG{p}{,} \PYG{n}{mtry} \PYG{o}{=} \PYG{n}{mtry}\PYG{p}{,}
                                  \PYG{n}{node\PYGZus{}size} \PYG{o}{=} \PYG{n}{node\PYGZus{}size}\PYG{p}{,} \PYG{n}{n\PYGZus{}basis} \PYG{o}{=} \PYG{n}{n\PYGZus{}basis}\PYG{p}{,}
                                  \PYG{n}{flambda} \PYG{o}{=} \PYG{n}{lambda\PYGZus{}param}\PYG{p}{)}

\PYG{c+c1}{\PYGZsh{} ... Same as RFCDE for prediction ...}
\end{Verbatim}
\noindent\rule{\textwidth}{0.5pt}

\bigskip

\subsection{\flexcode}
\label{app:flexcode}

\noindent\rule{\textwidth}{0.5pt}
\begin{Verbatim}[commandchars=\\\{\},codes={\catcode`\$=3\catcode`\^=7\catcode`\_=8}]
\PYG{k+kn}{import} \PYG{n+nn}{numpy} \PYG{k+kn}{as} \PYG{n+nn}{np}
\PYG{k+kn}{import} \PYG{n+nn}{flexcode}

\PYG{c+c1}{\PYGZsh{} Select regression method}
\PYG{k+kn}{from} \PYG{n+nn}{flexcode.regression\PYGZus{}models} \PYG{k+kn}{import} \PYG{n}{NN}

\PYG{c+c1}{\PYGZsh{} Parameters}
\PYG{n}{basis\PYGZus{}system} \PYG{o}{=} \PYG{l+s+s2}{\PYGZdq{}cosine\PYGZdq{}}  \PYG{c+c1}{\PYGZsh{} Basis system}
\PYG{n}{max\PYGZus{}basis} \PYG{o}{=} \PYG{l+m+mi}{31}           \PYG{c+c1}{\PYGZsh{} Maximum number of basis. If the model is not tuned,}
                         \PYG{c+c1}{\PYGZsh{} max\PYGZus{}basis is set as number of basis}

\PYG{c+c1}{\PYGZsh{} Regression Parameters}
\PYG{c+c1}{\PYGZsh{} If a list is passed for any parameter automatic 5\PYGZhy{}fold CV is used to}
\PYG{c+c1}{\PYGZsh{} determine the best parameter combination.}
\PYG{n}{params} \PYG{o}{=} \PYG{p}{\PYGZob{}}\PYG{l+s+s2}{\PYGZdq{}k\PYGZdq{}}\PYG{p}{:} \PYG{p}{[}\PYG{l+m+mi}{5}\PYG{p}{,} \PYG{l+m+mi}{10}\PYG{p}{,} \PYG{l+m+mi}{15}\PYG{p}{,} \PYG{l+m+mi}{20}\PYG{p}{]\PYGZcb{}}       \PYG{c+c1}{\PYGZsh{} A dictionary with method\PYGZhy{}specific regression parameters.}

\PYG{c+c1}{\PYGZsh{} Parameterize model}
\PYG{n}{model} \PYG{o}{=} \PYG{n}{flexcode}\PYG{o}{.}\PYG{n}{FlexCodeModel}\PYG{p}{(}\PYG{n}{NN}\PYG{p}{,} \PYG{n}{max\PYGZus{}basis}\PYG{p}{,} \PYG{n}{basis\PYGZus{}system}\PYG{p}{,} \PYG{n}{regression\PYGZus{}params}\PYG{o}{=}\PYG{n}{params}\PYG{p}{)}

\PYG{c+c1}{\PYGZsh{} Fit model \PYGZhy{} this will also choose the optimal number of neighbors `k`}
\PYG{n}{model}\PYG{o}{.}\PYG{n}{fit}\PYG{p}{(}\PYG{n}{x\PYGZus{}train}\PYG{p}{,} \PYG{n}{y\PYGZus{}train}\PYG{p}{)}

\PYG{c+c1}{\PYGZsh{} Tune model \PYGZhy{} Select the best number of basis}
\PYG{n}{model}\PYG{o}{.}\PYG{n}{tune}\PYG{p}{(}\PYG{n}{x\PYGZus{}validation}\PYG{p}{,} \PYG{n}{y\PYGZus{}validation}\PYG{p}{)}

\PYG{c+c1}{\PYGZsh{} Predict new densities on grid}
\PYG{n}{cde\PYGZus{}test}\PYG{p}{,} \PYG{n}{y\PYGZus{}grid} \PYG{o}{=} \PYG{n}{model}\PYG{o}{.}\PYG{n}{predict}\PYG{p}{(}\PYG{n}{x\PYGZus{}test}\PYG{p}{,} \PYG{n}{n\PYGZus{}grid}\PYG{o}{=}\PYG{n}{n\PYGZus{}grid}\PYG{p}{)}
\end{Verbatim}
\noindent\rule{\textwidth}{0.5pt}

\vspace{0.1cm}

\noindent\rule{\textwidth}{0.5pt}
\begin{Verbatim}[commandchars=\\\{\},codes={\catcode`\$=3\catcode`\^=7\catcode`\_=8}]
\PYG{n+nf}{library}\PYG{p}{(}\PYG{n}{FlexCoDE}\PYG{p}{)}

\PYG{c+c1}{\PYGZsh{} Parameters}
\PYG{n}{basis\PYGZus{}system} \PYG{o}{\PYGZlt{}\PYGZhy{}} \PYG{l+s}{\PYGZdq{}Cosine\PYGZdq{}}  \PYG{c+c1}{\PYGZsh{} Basis system}
\PYG{n}{max\PYGZus{}basis} \PYG{o}{\PYGZlt{}\PYGZhy{}} \PYG{l+m}{31}           \PYG{c+c1}{\PYGZsh{} Maximum number of basis.}

\PYG{c+c1}{\PYGZsh{} Fit and tune FlexCode via KNN regression using 4 cores}
\PYG{n}{fit} \PYG{o}{\PYGZlt{}\PYGZhy{}} \PYG{n}{FlexCoDE}\PYG{o}{::}\PYG{n+nf}{fitFlexCoDE}\PYG{p}{(}\PYG{n}{x\PYGZus{}train}\PYG{p}{,} \PYG{n}{y\PYGZus{}train}\PYG{p}{,} \PYG{n}{x\PYGZus{}validation}\PYG{p}{,} \PYG{n}{y\PYGZus{}validation}\PYG{p}{,}
           \PYG{n}{nIMax} \PYG{o}{=} \PYG{n}{max\PYGZus{}basis}\PYG{p}{,} \PYG{n}{regressionFunction.extra} \PYG{o}{=} \PYG{n+nf}{list}\PYG{p}{(}\PYG{n}{nCores} \PYG{o}{=} \PYG{l+m}{4}\PYG{p}{),}
           \PYG{n}{regressionFunction} \PYG{o}{=} \PYG{n}{FlexCoDE}\PYG{o}{::}\PYG{n}{regressionFunction.NN}\PYG{p}{,}
           \PYG{n}{system} \PYG{o}{=} \PYG{n}{basis\PYGZus{}system}\PYG{p}{)}

\PYG{c+c1}{\PYGZsh{} Predict new densities on a grid}
\PYG{n}{n\PYGZus{}points\PYGZus{}grid} \PYG{o}{\PYGZlt{}\PYGZhy{}} \PYG{l+m}{500}     \PYG{c+c1}{\PYGZsh{} Number of points in the grid}
\PYG{n}{cde\PYGZus{}test} \PYG{o}{\PYGZlt{}\PYGZhy{}} \PYG{n+nf}{predict}\PYG{p}{(}\PYG{n}{fit}\PYG{p}{,} \PYG{n}{x\PYGZus{}test}\PYG{p}{,} \PYG{n}{B} \PYG{o}{=} \PYG{n}{n\PYGZus{}points\PYGZus{}grid}\PYG{p}{)}

\PYG{c+c1}{\PYGZsh{}\PYGZsh{} Shortcut for FlexZBoost}

\PYG{c+c1}{\PYGZsh{} Fit and tune FlexZBoost}
\PYG{n}{fit} \PYG{o}{\PYGZlt{}\PYGZhy{}} \PYG{n}{FlexCoDE}\PYG{o}{::}\PYG{n+nf}{FlexZBoost}\PYG{p}{(}\PYG{n}{x\PYGZus{}train}\PYG{p}{,} \PYG{n}{y\PYGZus{}train}\PYG{p}{,} \PYG{n}{x\PYGZus{}validation}\PYG{p}{,} \PYG{n}{y\PYGZus{}validation}\PYG{p}{,}
                            \PYG{n}{nIMax} \PYG{o}{=} \PYG{n}{max\PYGZus{}basis}\PYG{p}{,} \PYG{n}{system} \PYG{o}{=} \PYG{n}{basis\PYGZus{}system}\PYG{p}{)}

\PYG{c+c1}{\PYGZsh{} Predict new densities on a grid}
\PYG{n}{n\PYGZus{}points\PYGZus{}grid} \PYG{o}{\PYGZlt{}\PYGZhy{}} \PYG{l+m}{500}     \PYG{c+c1}{\PYGZsh{} Number of points in the grid}
\PYG{n}{cde\PYGZus{}test} \PYG{o}{\PYGZlt{}\PYGZhy{}} \PYG{n+nf}{predict}\PYG{p}{(}\PYG{n}{fit}\PYG{p}{,} \PYG{n}{x\PYGZus{}test}\PYG{p}{,} \PYG{n}{B} \PYG{o}{=} \PYG{n}{n\PYGZus{}points\PYGZus{}grid}\PYG{p}{)}
\end{Verbatim}
\noindent\rule{\textwidth}{0.5pt}

\newpage

\subsection{\cdetools}
\label{app:cdetools}

\noindent\rule{\textwidth}{0.5pt}
\begin{Verbatim}[commandchars=\\\{\},codes={\catcode`\$=3\catcode`\^=7\catcode`\_=8}]
\PYG{k+kn}{import} \PYG{n+nn}{cdetools}

\PYG{n}{cde\PYGZus{}test}   \PYG{c+c1}{\PYGZsh{} numpy matrix of conditional density evaluations on a grid}
           \PYG{c+c1}{\PYGZsh{} \PYGZhy{} each row is an observation, each column is a grid point}
\PYG{n}{y\PYGZus{}grid}     \PYG{c+c1}{\PYGZsh{} the grid at which cde\PYGZus{}test is evaluated}
\PYG{n}{y\PYGZus{}true}     \PYG{c+c1}{\PYGZsh{} the observed y values}

\PYG{c+c1}{\PYGZsh{} Calculate the cde\PYGZus{}loss}
\PYG{n}{cde\PYGZus{}loss}\PYG{p}{(}\PYG{n}{cde\PYGZus{}test}\PYG{p}{,} \PYG{n}{y\PYGZus{}grid}\PYG{p}{,} \PYG{n}{y\PYGZus{}test}\PYG{p}{)}

\PYG{c+c1}{\PYGZsh{} Calculate PIT values}
\PYG{n}{cdf\PYGZus{}coverage}\PYG{p}{(}\PYG{n}{cde\PYGZus{}test}\PYG{p}{,} \PYG{n}{y\PYGZus{}grid}\PYG{p}{,} \PYG{n}{y\PYGZus{}test}\PYG{p}{)} 

\PYG{c+c1}{\PYGZsh{} Calculate HPD values}
\PYG{n}{hpd\PYGZus{}coverage}\PYG{p}{(}\PYG{n}{cde\PYGZus{}test}\PYG{p}{,} \PYG{n}{y\PYGZus{}grid}\PYG{p}{,} \PYG{n}{y\PYGZus{}test}\PYG{p}{)}
\end{Verbatim}
\noindent\rule{\textwidth}{0.5pt}

\vspace{0.1cm}

\noindent\rule{\textwidth}{0.5pt}
\begin{Verbatim}[commandchars=\\\{\},codes={\catcode`\$=3\catcode`\^=7\catcode`\_=8}]
\PYG{n+nf}{library}\PYG{p}{(}\PYG{n}{cde\PYGZus{}tools}\PYG{p}{)}

\PYG{n}{cde\PYGZus{}test}   \PYG{c+c1}{\PYGZsh{} matrix of conditional density evaluations on a grid}
           \PYG{c+c1}{\PYGZsh{} \PYGZhy{} each row is an observation, each column is a grid point}
\PYG{n}{y\PYGZus{}grid}     \PYG{c+c1}{\PYGZsh{} the grid at which cde\PYGZus{}test is evaluated}
\PYG{n}{y\PYGZus{}true}     \PYG{c+c1}{\PYGZsh{} the observed y values}

\PYG{c+c1}{\PYGZsh{} Calculate the cde\PYGZus{}loss}
\PYG{n+nf}{cde\PYGZus{}loss}\PYG{p}{(}\PYG{n}{cde\PYGZus{}test}\PYG{p}{,} \PYG{n}{y\PYGZus{}grid}\PYG{p}{,} \PYG{n}{y\PYGZus{}test}\PYG{p}{)}

\PYG{c+c1}{\PYGZsh{} Calculate PIT values}
\PYG{n+nf}{cdf\PYGZus{}coverage}\PYG{p}{(}\PYG{n}{cde\PYGZus{}test}\PYG{p}{,} \PYG{n}{y\PYGZus{}grid}\PYG{p}{,} \PYG{n}{y\PYGZus{}test}\PYG{p}{)}

\PYG{c+c1}{\PYGZsh{} Calculate HPD values}
\PYG{n+nf}{hpd\PYGZus{}coverage}\PYG{p}{(}\PYG{n}{cde\PYGZus{}test}\PYG{p}{,} \PYG{n}{y\PYGZus{}grid}\PYG{p}{,} \PYG{n}{y\PYGZus{}test}\PYG{p}{)}
\end{Verbatim}
\noindent\rule{\textwidth}{0.5pt}

\section{Point Estimate Photo-Z Redshifts (Univariate Response with Multivariate Data)}
\label{sec:app_pointestimates}

Plots of point estimates can sometimes be a useful qualitative diagnostic of \pz\ code performance in that they may quickly identify a bad model. 
For example, the marginal model would always lead to the same point estimate and hence a horizontal line in the plot. 
Figure~\ref{fig:setting1_examples_prediction} shows point \pz\ predictions versus the true redshift for the \pz\ codes in Example 4.1.

\begin{figure}[htbp]
\centering
\label{fig:teddy-examples-predictions}
\includegraphics[width=0.9\linewidth]{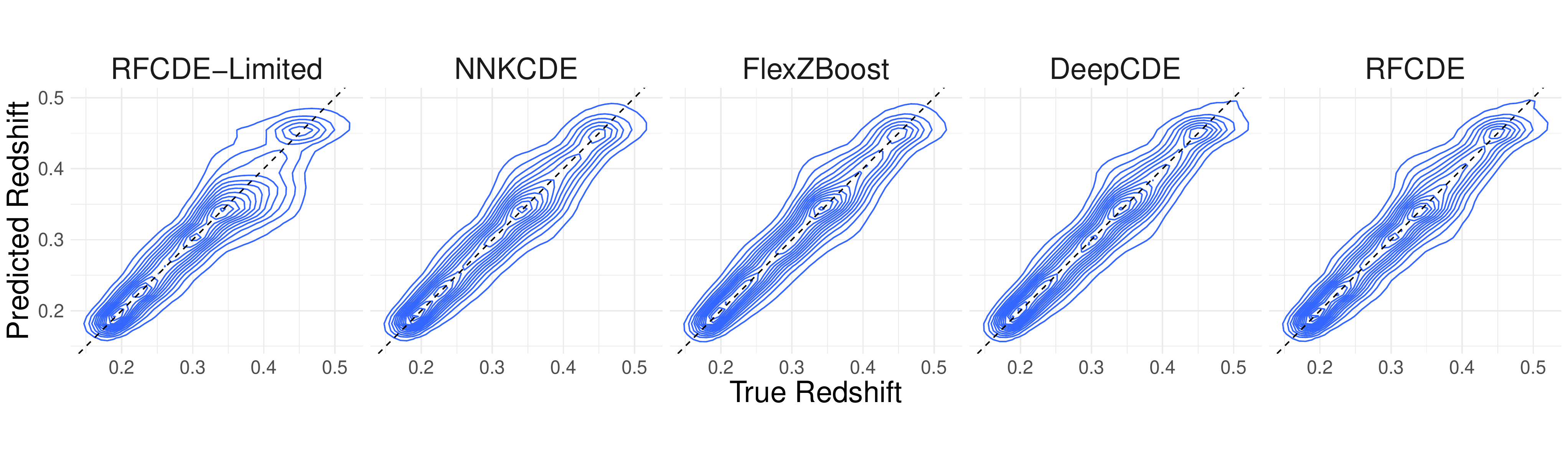}
\caption{Predicted photometric versus true spectroscopic redshift.
Point estimate \pz s are derived from the CDEs by taking the mode of each density.
Predictions from the marginal, not shown here, always predicts the mode of the marginal distribution, which here is  $z = 0.3417$.
 \label{fig:setting1_examples_prediction}}
\end{figure}

The second moment of the distribution and outlier fraction (OLF) are also commonly used diagnostic tools for \pz\ estimators.
Table \ref{tab:photoz_metrics_teddy} lists the $\sigma_f$, $\sigma_{\rm nmad}$ and OLF point estimate metrics from \citet{dahlen2013photoz_assessment}. 
These metrics are defined as follows:
\begin{align}
    \sigma_f &= {\rm rms}\left(\frac{\Delta z}{1 + z_{\rm spec}} \right) = \sqrt{\frac{1}{n} \sum_{i=1}^n \left( \frac{\Delta z_i}{1 + z_{\rm spec,i}} \right)^2} \\
    \sigma_{\rm nmda} &= 1.48 \cdot {\rm median} \left( \frac{\abs{ \Delta z}}{1 + z_{\rm spec}} \right) \\
    {\rm OLF} &= \frac{1}{n} \sum_{i=1}^n \mathbb{I}\left(\frac{\abs{ \Delta z_i}}{1 + z_{\rm spec, i}} > 0.15\right)
\end{align}
where $z_{\rm spec}$ indicates the spectroscopic redshift and $\Delta z$ is the difference between the predicted and spectroscopic redshift. 
According to the point estimate metrics, \rfcde\ and \deepcde\ perform the best on the Teddy data, followed by \flexzboost\ and \nnkcde.
These results are consistent with the CDE loss rankings in Figure~\ref{tab:teddy-cde-loss}.

\begin{table}[htbp]
  \label{tab:photoz_metrics_teddy}
  \caption{\small Performance in photometric redshift predictions for different methods. 
  }
  \centering
  \begin{tabular}{lrrr}
    Method & $\sigma_f$ & $\sigma_{\rm nmad}$ & OLF (\%)\\
    \hline
    Marginal & 0.0049 & 0.0740 & 1.018 \\
    RFCDE-Limited & 0.0010 & 0.0205 & 0.604 \\
    NNKCDE & 0.0008 & 0.0180 & 0.335 \\
    FlexZBoost & 0.0008 & 0.0169 & 0.362 \\
    DeepCDE & 0.0007 & 0.0164 & 0.359 \\
    RFCDE & 0.0007 & 0.0169 & 0.345
  \end{tabular}
\end{table}


\section{Loss Function and Model Assumption Equivalence} 
\label{app:cdeloss_equiv}

In this section we connect the \deepcde\ loss in Equation \ref{eq:cde-loss} to the \flexcode\ model setup in Equation \ref{eq:flexcode}.

\begin{Lemma}
Let $\betavec = \{\beta_i(\mathbf{x})\}_{i=1}^B$ be the coefficients of the \flexcode\ basis expansion in Equation \ref{eq:flexcode}. 
Minimizing the CDE loss in Equation \ref{eq:cde-loss} is equivalent to minimizing the mean squared errors of the basis expansion coefficients, i.e., 
\begin{equation}
    \min_{\hat{\betavec} \in \mathbb{R}^B} \int_\mathcal{X} \int_\mathcal{Y} (\widehat{p}(y | \x) - p(y | \x))^{2} dy dP(\x) \Longleftrightarrow \min_{\hat{\betavec} \in \mathbb{R}^B} \mathbb{E}_{\x} \left[ \norm{\hat{\betavec}(\x) - \betavec(\x)}^2 \right] .
    \label{eq::lemma_deepCDE}
\end{equation}
\end{Lemma}

\begin{proof}
We prove the statement by showing that the two minimization problems are equivalent.\\

First, considering the LHS of Equation~\ref{eq::lemma_deepCDE}, we have that:

\begin{align*}
    & \min_{\hat{\betavec} \in \mathbb{R}^B} \int_\mathcal{X} \int_\mathcal{Y} (\widehat{p}(y | \x) - p(y | \x))^{2} dy dP(\x) \\
    &\Longleftrightarrow \min_{\hat{\betavec} \in \mathbb{R}^B} \int_\mathcal{X} \int_\mathcal{Y} \widehat{p}(y | \x)^{2} dy dP(\x) - 2 \int_\mathcal{X} \int_\mathcal{Y} \widehat{p}(y | \x) p(\x,y) d\x dy \\
    &= \min_{\hat{\betavec} \in \mathbb{R}^B} \int_\mathcal{X} \int_\mathcal{Y} \left(\sum_{i=1}^B \hat{\beta}_i(\x) \phi_i(y) \right)^{2} dy dP(\x) - 2 \int_\mathcal{X} \int_\mathcal{Y} \left(\sum_{i=1}^B \hat{\beta}_i(\x) \phi_i(y) \right) p(y | \x) dP(\x)dy \\
    &= \min_{\hat{\betavec} \in \mathbb{R}^B} \int_\mathcal{X} \sum_{i,j=1}^B \hat{\beta}_i(\x) \hat{\beta}_j(\x) \int_\mathcal{Y}  \phi_i(y) \phi_j(y) dy dP(\x) - 2 \int_\mathcal{X} \int_\mathcal{Y} \left(\sum_{i=1}^B \hat{\beta}_i(\x) \phi_i(y) \right) \left(\sum_{j=1}^B \beta_j(\x) \phi_j(y) \right) dP(\x) dy \\
    &= \min_{\hat{\betavec} \in \mathbb{R}^B} \int_\mathcal{X} \sum_{i=1}^B \hat{\beta}^2_i(\x) dP(\x) - 2 \int_\mathcal{X} \sum_{i=1}^B \hat{\beta}_i(\x) \beta_i(\x)  dP(\x) 
    = \min_{\hat{\betavec} \in \mathbb{R}^B} \sum_{i=1}^B  \mathbb{E}_{\x} \left[ \hat{\beta}^2_i(\x) - 2 \hat{\beta}_i(\x) \beta_i(\x)\right],
\end{align*}
where the second equality follows from the fact that the set $\{\phi_i(y)\}_{i=1}^B$ is part of an orthonormal basis, that is,
\begin{equation*}
    \int_\mathcal{Y} \phi_i(y) \phi_j(y) dy = \delta_{ij} = \begin{cases} 1 & i=j \\ 0 & i \neq j  \end{cases}.
\end{equation*}\\

Next, the RHS of Equation~\ref{eq::lemma_deepCDE} reduces to:
\begin{align*}
\min_{\hat{\betavec} \in \mathbb{R}^B} \mathbb{E}_{\x} \left[ \norm{\hat{\betavec}(\x) - \betavec(\x)}^2 \right] = 
\min_{\hat{\betavec} \in \mathbb{R}^B} \sum_{i=1}^B \mathbb{E}_{\x} \left[ \left(\hat{\beta}_i(\x) - \beta_i(\x) \right)^2  \right] 
\Longleftrightarrow \min_{\hat{\betavec} \in \mathbb{R}^B} \sum_{i=1}^B \mathbb{E}_{\x} \left[ \hat{\beta}^2_i(\x) - 2\hat{\beta}_i(\x)\beta_i(\x) . \right]
\end{align*}

\end{proof}

In \deepcde, we implement an especially simple expression for the loss function by inserting the orthogonal basis expansion (Equation~\ref{eq::expansion_DeepCDE}) into the CDE loss (Equation~\ref{eq:cde-loss}) to yield
\begin{eqnarray}
    L(p, \widehat{p}) &=&  \int \sum_{j=1}^B \hat{\beta}^2_j(\x) dP(\x) - 2 \int \int \sum_{j=1}^B \hat{\beta}_j(\x) \phi_j(y)  dP(\x,y) \nonumber\\ 
    &\approx&
    \frac{1}{n} \sum_{i=1}^n \left( \sum_{j=1}^B  \hat{\beta}^2_j(\mathbf{x}_{i}^{te}) - 2 \sum_{j=1}^B  \hat{\beta}_j(\mathbf{x}_{i}^{te}) \phi_j(y_{i}^{te}) \right).
\end{eqnarray}
The last expression represents our empirical CDE loss on validation data $\left\{(\mathbf{x}_{i}^{te}, y_{i}^{te})\right\}_{i=1}^{n}$
and is easy to compute.

\vspace{0.33cm}

\section{Accounting for Covariate Shift in Photometric Redshift Estimation via Importance Weights}
\label{sec:cov_shift_photoz}

In the photo-$z$ application of Section \ref{sec:univ_response_mult_data} we have for simplicity assumed the distributions of the spectroscopic and photometric data to be the same. 
In this section we present a solution by  \cite{izbicki2017photo} and \cite{freeman2017unified}  for a setting where the spectroscopically labeled galaxies and the photometric ``target'' galaxies of interest lie in different regions of the color space.
Under so-called {\em covariate shift} 
the the (marginal) distributions $\ptr(\x)$ and $\pte(\x)$ of the spectroscopic and photometric data, respectively, are different --- but their  conditional distributions $\ptr(z|\x)$ and $\pte(z|\x)$ are assumed to be the same. In other words: 
 $$ \ptr(\x) \neq \pte(\x)\quad, \quad \ptr(z|\x) = \pte(z|\x) \implies \ptr(z,\x) \neq \pte(z,\x). $$

Covariate shift arises from missing at random (MAR) bias, which for example occurs when the probability that a galaxy is labeled with a spectroscopic redshift depends only on its photometry and possibly other characteristics $\x$ but not on its true redshift $z$. 
In such a setting the goal is to evaluate the performance on the photometric ``target'' data of interest. 
Hence, in terms of loss functions, our goal is to minimize a CDE loss computed on the {\em photometric}, and not the spectroscopic, data as in 
\begin{equation}
\label{eq:mod_cde_loss_cs}
    L(\widehat{p},p) =\iint \left(\widehat{p}(z|\x)-p(z|\x)\right)^2 \pte(\x) d\x dz ,
\end{equation}
where the weight function $\pte(\x)$ will enforce a good fit in the regions of color or feature space where most of the photometric data reside.
The challenge is how to compute the above loss for the photometric data of interest, as the true redshift $z$, by construction, is only known for the spectroscopic set. 
A solution proposed in the above-mentioned papers, is to rewrite the loss by introducing so-called {\em importance weights} $\beta(\x) := \pte(\x)/\ptr(\x)$.
An example of how to estimate such weights can be found in \cite{izbicki2017photo}, \cite{Sugiyama2007ImportanceWeights}, \cite{Gretton2009CovariateShift} and \cite{Kremer2015ImportanceWeightsNN}.
Once the importance weights have been estimated, one can rewrite an estimate of the loss in \ref{eq:mod_cde_loss_cs} as an expression that can be computed on held-out data  without knowing the redshift $z$ of the (photometric) target data: 
\begin{align}
\label{eq:emp_cde_loss_cs}
    \hat{L}(\hat{p}, p) = \frac{1}{n_{\rm photo}}\sum_{j=1}^{n_{\rm photo}}\int \widehat{p}^2\left(z | \x_j^{\rm photo}\right)dz-2\frac{1}{n_{\rm spec}}\sum_{k=1}^{n_{\rm spec}}\widehat{p}\left(z^{\rm spec}_k|\x^{\rm spec}_k\right)\widehat{\beta}\left(\x^{\rm spec}_k\right) ,
\end{align}
where $\{z_k^{\rm spec}, \x_k^{\rm spec} \}_{k=1}^{n_{\rm spec}}$ and $\{z_j^{\rm photo}, \x_j^{\rm photo} \}_{j=1}^{n_{\rm photo}}$ represent the spectroscopic and photometric data, respectively, in the held-out set.
Furthermore, as noted in \cite{izbicki2017photo} (Section 5), one can often drastically improve CDE performance by just restricting the spectroscopic training data to regions where the unlabeled photometric data reside. 
A simple procedure for identifying such training examples is to compute importance weights for all labeled examples, and then replace the examples in the training set with $\hat{\beta}(\x) = 0$ with new labeled examples that have  $\hat{\beta}(\x) \neq 0$.
Effectively, this will create a spectroscopic training set that resides in the same part of the feature space as the unlabeled photometric test data (when that is possible).

Figure \ref{fig:covariate_shift} (adapted from \citet{izbicki2017photo}) shows the performance of three versions of the \flexcode-series.
The difference between these approaches is in how they are tuned.
The first result ("no reweighting") corresponds to the standard version of \flexcode\ that does not take into account covariate shift and hence is tuned on held-out validation data according to Equation \eqref{eq:estimated-cde-loss}. 
The second estimator ("removing zeros") is also tuned according to Equation \eqref{eq:estimated-cde-loss}, but on a held-out sample that has no training examples in regions with $\hat \beta (\x)=0$.
For the last estimator ("reweighting"), tuning is performed according to Equation \eqref{eq:emp_cde_loss_cs}. 
The plot shows that the simple approach of removing galaxies with zero estimated weights already improves the estimators, but reweighting the loss by $\hat \beta$ 
can yield even better performance.

\begin{figure}[htbp]
\centering
\includegraphics[width=0.475\linewidth]{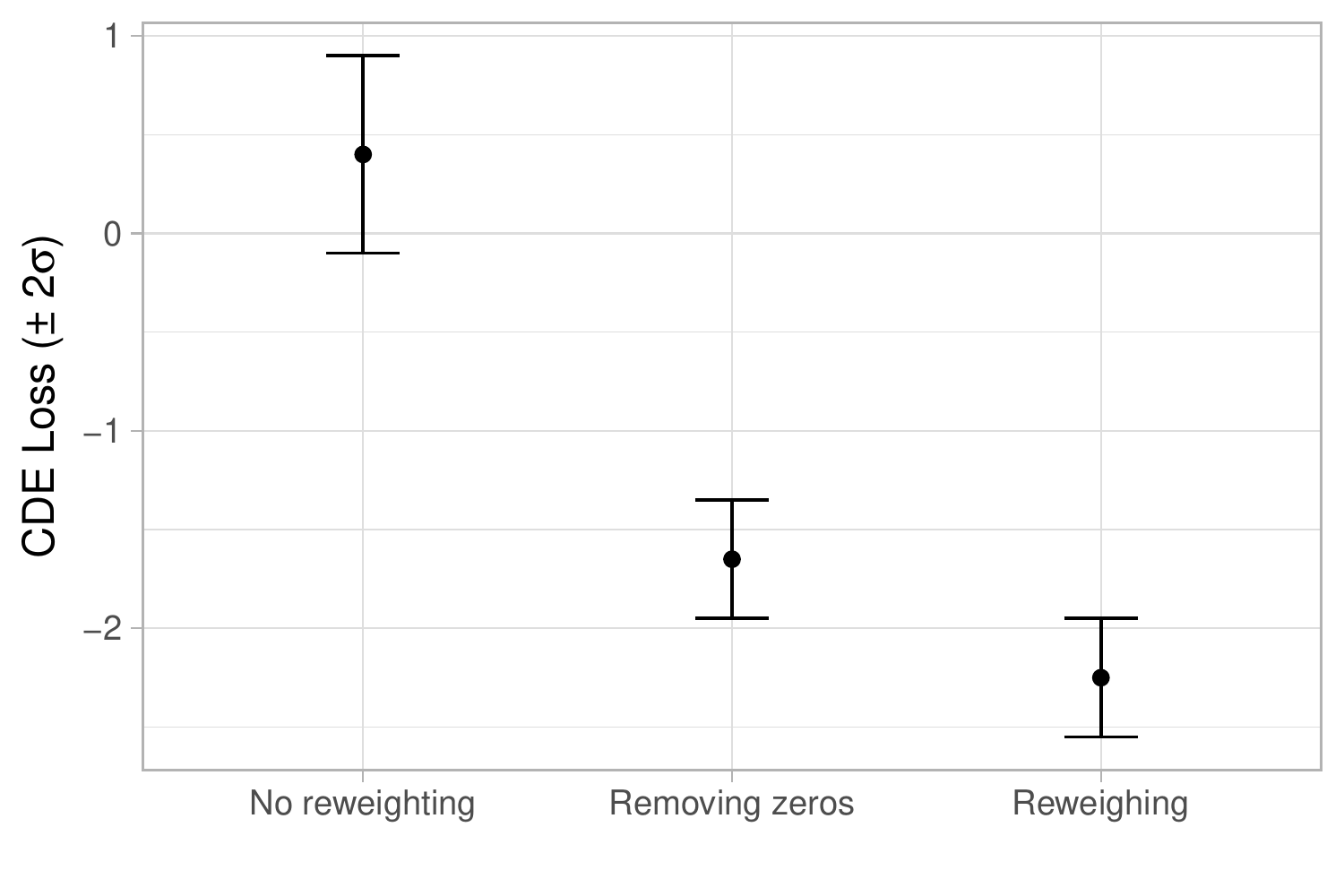}
\caption{Performance of three versions of 
\flexcode-Series to the photometric data used in
\citep{izbicki2017photo}. Left: estimator that ignores selection bias;
middle: removing and replacing galaxies with  $\hat  \beta(\x)=0$ in the training set, but otherwise not including the importance weights in the computations; 
right: reweighting galaxies by $\hat \beta(\x)$.
(Adapted from \citep{izbicki2017photo}). }
\label{fig:covariate_shift}
\end{figure}

\end{document}